\documentclass[twoside]{article}
\usepackage[accepted]{aistats2024}

\usepackage[utf8]{inputenc} 
\usepackage[T1]{fontenc}    
\usepackage{url}            
\usepackage{booktabs}       
\usepackage{amsfonts}       
\usepackage{nicefrac}       
\usepackage{microtype}      
\usepackage[dvipsnames]{xcolor}        
\usepackage{amsmath,amsfonts,amssymb, amsthm}
\usepackage{stmaryrd}
\usepackage{setspace}
\usepackage{subcaption}
\usepackage{graphicx}
\usepackage{wrapfig}
\usepackage{algorithm}
\usepackage{algpseudocode}
\usepackage[font=small,skip=0pt]{caption}
\usepackage{tabularray}
\usepackage{adjustbox}
\usepackage{comment}
\usepackage{thmtools}
\usepackage{etoolbox}
\usepackage[toc,page,header]{appendix}
\usepackage{minitoc}
\usepackage[round]{natbib}
\usepackage{multirow}
\usepackage{float}
\usepackage{dblfloatfix}
\usepackage{etoolbox}

\definecolor{mydarkblue}{rgb}{0,0.08,0.45}
\usepackage[colorlinks=true, linkcolor=mydarkblue, citecolor=mydarkblue]{hyperref}      

\newtheorem{thm}{Theorem}
\newtheorem{pro}{Proposition}
\newtheorem{cor}{Corollary}
\newtheorem{lem}{Lemma}

\newtheorem{rem}{Remark}
\newtheorem{defn}{Definition}

\newtheorem{assumption}{Assumption}
\newtheorem*{statement}{Statement}

\newcommand{\ceil}[1]{\left\lceil #1 \right\rceil}

\newcommand{\Identity}{\mathrm{I}}
\newcommand{\E}{\mathbb{E}}
\newcommand{\R}{\mathbb{R}}

\allowdisplaybreaks
\usepackage{array}
    \newcolumntype{P}[1]{>{\centering\arraybackslash}p{#1}}
    \newcolumntype{M}[1]{>{\centering\arraybackslash}m{#1}}

\newcommand{\card}[1]{\left\lvert#1\right\rvert}
\newcommand{\expp}{ \mathbb{E}}

\newcommand{\calS}{\mathcal{S}}

\usepackage{enumitem}
\setlist{leftmargin=5.5mm}

\newcommand{\norm}[1]{\left\lVert#1\right\rVert}
\newcommand{\cleandata}{$\{(\tilde{s}^\tau_h,\tilde{a}^\tau_h,\tilde{b}^\tau_h,\tilde{r}^\tau_h,\tilde{s}^\tau_{h+1})\}^{K,H}_{\tau=1,h=1}$}
\newcommand{\data}{$\{(s^\tau_h,a^\tau_h,b^\tau_h,r^\tau_h,s^\tau_{h+1})\}^{K,H}_{\tau=1,h=1}$}
\newcommand{\optgap}{\text{SubOpt}}

\begin{document}

%

%
\runningauthor{Andi Nika, Debmalya Mandal, Adish Singla and Goran Radanovic}

\twocolumn[

\aistatstitle{Corruption-Robust Offline Two-Player Zero-Sum Markov Games}

\aistatsauthor{ Andi Nika \And Debmalya Mandal$^\dagger$ \And Adish Singla \And Goran Radanovic}
\aistatsaddress{ MPI-SWS \And University of Warwick \And MPI-SWS \And MPI-SWS } ]

\doparttoc 
\faketableofcontents 

\begin{abstract}
We study data corruption robustness in offline two-player zero-sum Markov games. Given a dataset of realized trajectories of two players, an adversary is allowed to modify an $\epsilon$-fraction of it. The learner's goal is to identify an approximate Nash Equilibrium policy pair from the corrupted data.  We consider this problem in linear Markov games under different degrees of data coverage and corruption. We start by providing an information-theoretic lower bound on the suboptimality gap of any learner. Next, we propose robust versions of the Pessimistic Minimax Value Iteration algorithm \citep{zhong2022pessimistic}, both under coverage on the corrupted data and under coverage only on the clean data, and show that they achieve (near)-optimal suboptimality gap bounds with respect to $\epsilon$. We note that we are the first to provide such a characterization of the problem of learning approximate Nash Equilibrium policies in offline two-player zero-sum Markov games under data corruption.
\end{abstract}


\section{INTRODUCTION}

\begin{figure*}[!t]
    \centering    
    \includegraphics[width=0.9\textwidth]{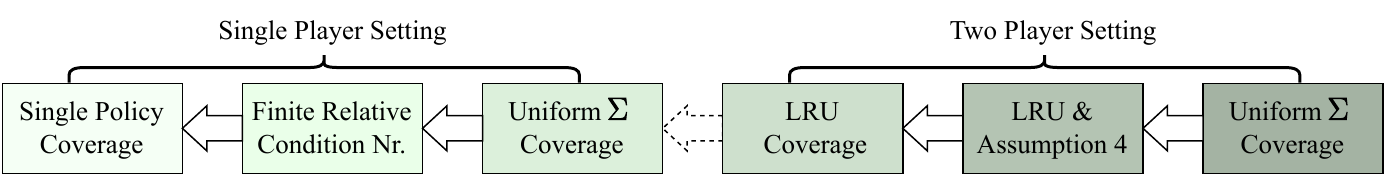}\vspace{0.5em}
    \caption{Relationship between coverage assumptions. The minimal coverage requirements are single policy coverage and LRU coverage for the single-player and two-player settings, respectively.  Arrows stand for implications. The middle (dashed) arrow denotes the restriction from the two-player to the single-player setting: when fixing the second player's policy, the LRU coverage assumption reduces to the uniform coverage assumption.}
    \label{fig:coverage_relationship}
\end{figure*}

Some of the most successful applications of Multi-agent Reinforcement Learning (MARL) are in competitive game-playing \citep{silver2017mastering, berner2019dota}, where we have a model of the environment that we can use for training purposes. %
Given that many real-world multi-agent applications, such as autonomous driving \citep{pan2017agile} or healthcare \citep{wang2018supervised}, do not have readily available simulators, there has recently been a growing interest in studying offline settings, where offline data is used to derive agents' policies. 
Since a dynamic exploration of the environment is impossible, state-of-the-art (SOTA) algorithms use the paradigm of \textit{pessimism in the face of uncertainty}  to derive these policies \citep{jin2021pessimism}.
Moreover, these works typically assume that the data is coming from a latent distribution with ``nice'' properties.

In practice, however,  datasets may be subject to adversarial attacks that corrupt data points and can significantly impact the performance of the learning process.  
Such security threats have already been explored in single-agent RL, where prior work has proposed corruption robust algorithms \citep{zhang2022corruption}. However, these results do not directly translate to MARL due to the intricacies of multi-agent settings. For instance, the learning objective in these settings requires a more complex solution concept of learning a Nash Equilibrium (NE) policy pair for agents instead of simply learning a near-optimal policy for an agent. In this work, we initiate the study of corruption robust algorithms for learning equilibrium policies in offline MARL. More specifically, we focus on two-agent zero-sum Markov games and consider the following research question: 

\looseness-1\textit{Can we design algorithms that approximately solve offline two-player zero-sum Markov games under data corruption?}

To effectively answer this question, we need to account for another crucial factor in offline learning, that is, the quality of the collected data, which drastically affects the quality of the learned policy. It is thus common practice to assume that the collected data \textit{covers} at least some trajectories of interest. It turns out that the necessary coverage assumptions for solving the offline single-player problem are not enough to solve the offline two-player problem (see Figure \ref{fig:coverage_relationship}). Thus, stronger assumptions are required, i.e., the so-called \textit{Low Relative Uncertainty} (LRU) assumption. This problem is exacerbated by the presence of corruption in our setting. If good coverage in the clean setting seems natural, supposing that the data has been collected by a \textit{good enough} policy, such an assumption is no longer guaranteed when a potentially malicious adversary intervenes in the data. 
\def\a{$\widetilde{O}\Bigl( H^2SAB\epsilon + H^2K^{-1/2}(SAB)^{3/2} \Bigr)$}
\def\b{$\widetilde{O}\Bigl( H^2d\epsilon + H^2K^{-1/2}d^{3/2}\Bigr)$ }
\def\c{$\widetilde{O}\Bigl( H^2SAB\epsilon + H^{3/2}K^{-1/2}f(SAB)  \Bigr)$}
\def\d{$\widetilde{O}\Bigl( H^2d\epsilon + H^{3/2}f(d)K^{-1/2} \Bigr)$}
\def\e{$\widetilde{O}\Bigl( H^2d^{3/2}\sqrt{\epsilon} + H^2K^{-1/2}d^{3/2}\Bigr)$}
\def\f{$\widetilde{O}\Bigl( H^2d^{3}\epsilon + H^2K^{-1/2}d^{3}\Bigr)$}

\begin{table*}[!b]
\renewcommand{\arraystretch}{1.7}
    \centering
    \scalebox{0.9}
{
\begin{tabular}{l c c c c}\toprule
        Coverage & Covariates & Algorithm & Suboptimality Gap & Result \\
        \midrule
        Uniform $\Sigma$ on corrupted data & Corrupted & R-PMVI & \d & [Theorem \ref{thm:linear_uniform_coverage}] \\  
        LRU on corrupted data & Clean & S-PMVI &  \b & [Theorem \ref{thm:linear_case}]  \\ 
        LRU on clean data & Corrupted & S-PMVI & \e & [Theorem \ref{thm:general_result}] \\
         LRU on clean data  \& A.\ref{asm:features_lower_bound}  & Corrupted & S-PMVI & \f & [Theorem \ref{thm:second_general_result}] \\
        \bottomrule
\end{tabular}
}
\vspace{0.1in}
\caption{Summary of our results under Low Relative Uncertainty and uniform $\Sigma$-coverage assumptions (see Assumptions \ref{asn:uniform-phi-coverage} and \ref{asm:low_relative_uncertainty} for definitions) on clean or corrupted data, and different corruption levels of the feature covariance matrix. Here $\epsilon$ denotes the corruption level, $K$ denotes the number of trajectories contained in the data, and $f(x)$ denotes a polynomial function of $x$. The Covariates column refers to whether the state-action part of the data tuple is corrupted or not. We have omitted linear dependence on noise variance $\gamma^2$ of the rewards for ease of presentation. We point the reader to the relevant results for a detailed description. We note that the universal lower bound for the offline two-player zero-sum MG setting is $\Omega (Hd\epsilon)$. }
\label{tab:summary}
\end{table*}
Motivated by these observations, we study the problem of corruption in the offline two-player setting under various assumptions. First, we tackle corruption under the minimal LRU coverage and the more relaxed uniform $\Sigma$-coverage (defined in Assumption \ref{asn:uniform-phi-coverage}) assumptions on the corrupted data. Furthermore, we also consider the more difficult setting where the minimal LRU coverage holds only on the clean dataset. We build upon recent techniques for the offline two-player zero-sum setting \citep{zhong2022pessimistic} and propose robust versions of their method. We tackle the corruption problem by using two setting-specific robust estimators and carefully designing new bonus terms that capture the additional estimation errors coming from corruption. 
More concretely, our main results and contributions are summarized below (also, see Table~\ref{tab:summary}):
\begin{enumerate}[label={\Roman*.}]
\setlength{\leftmargin}{-3pt}
\setlength\itemsep{-0.1em}
\item \textbf{Lower bound.} First, we formulate the problem of data corruption in offline two-player zero-sum Markov games. We prove an information-theoretic lower bound of $\Omega (Hd\epsilon)$ on the suboptimality gap of any algorithm that uses a corrupted dataset where $H$ is the episode length, $d$ is the dimension, and $\epsilon$ is the corruption level.
\item \textbf{Uniform $\Sigma$ and corrupted covariates.} Next, we consider the corruption problem under the Uniform $\Sigma$-coverage assumption on the corrupted data. We propose R-PMVI that uses a Robust Least Squares estimator \citep{zhang2022corruption} and show that it incurs a near-optimal bound on the error coming from corruption with high probability,  showing an improvement on the single-player bounds of \cite{zhang2022corruption} by a factor of $H$.
\item \textbf{LRU and clean covariates.} Furthermore, we consider the LRU coverage on the corrupted data. We additionally assume that corruption is done only on the reward and next state part of the data tuples, i.e., assuming clean covariates. In this setting, we propose S-PMVI that uses the Spectrally Regularized Alternating Minimization algorithm \citep{chen2022online} oracle as a robust estimator and provide a high probability near-optimal bound on its suboptimality gap. We note that, under no corruption, we recover the SOTA bounds of \cite{zhong2022pessimistic}, while obtaining similar rates of $\epsilon$ as in the single agent SOTA \citep{zhang2022corruption}.
\item \textbf{LRU on the clean data.} Finally, we consider the most restrictive setting where no guarantees of coverage on the corrupted data are given but only make the  necessary LRU assumption on the clean data. In this setting, we propose a new bonus term for the S-PMVI algorithm, that takes into account the additional error terms. Our method yields $O(d^{3/2}\sqrt{\epsilon})$ bounds on the suboptimality gap, similar to the single-player setting. Moreover, under an additional mild assumption on the feature space, we are able to recover the optimal $O(\epsilon)$ rate, at the cost of an additional $O(d^{3/2})$ factor. 
\end{enumerate}

Furthermore, we observe that convergence to Nash equilibria, without knowledge of $\epsilon$ and uniform coverage, is impossible in the offline two-player setting, which is a direct implication of the single-player setting \citep{zhang2022corruption}. Finally, we provide a comprehensive discussion on the relationship between used coverage assumptions and other similar assumptions found in the offline MARL literature and position them relative to each other.

\section{PRELIMINARIES}\label{sec:preliminaries}

\textbf{Notation.} We begin this section by introducing the notation to be used throughout the paper. As usual, $I$ denotes the identity matrix, $\norm{\cdot}_2$ denotes the Euclidean norm, $\norm{\cdot}_A$ denotes the Mahalanobis norm given square matrix $A$, and $\Delta(\mathcal{X})$ denotes the probability simplex on set $\mathcal{X}$. When $\widetilde{O}$ is used, any polylogarithmic terms are omitted. Furthermore, $[H]$ denotes the set of natural numbers up to and including $H$, $\langle \cdot, \cdot \rangle$ denotes the inner product, $\mathbf{1}\{ \cdot\}$ denotes the indicator function, and $\succeq$ denotes the Loewner order, where $A \succeq B$ is equivalent to $A-B$ being positive semi-definite.  Finally, we define $\Pi_h(x)=\min \{ h, \max \{ x, 0\} \}$. 

\subsection{Two-player Zero-sum Markov Games}

Let $\mathcal{G} = (\mathcal{S},\mathcal{A},\mathcal{B},p,r,H)$ be a finite-horizon zero-sum Markov game between two players, one of which is trying to maximize the total reward and the other is trying to minimize it. Here $\mathcal{S}$ denotes the state space with $S$ states\footnote{We only introduce the state space cardinality notation for convenience. Note that in linear Markov games, $S$ may be intractably large.}, $\mathcal{A}$ is the action space of the first player with $A$ actions, $\mathcal{B}$ is the action space of the second player with $B$ actions, $p = (p_1,\ldots, p_H)$ where $p_h \in \mathbb{R}^{SAB \times S}$ $\forall h \in [H]$ denote the transition kernels, $r=(r_1,\ldots,r_H)$ where $r_h\in \Delta(\mathbb{R})^{SAB}$ $\forall h\in[H]$ are $\gamma^2$-subGaussian rewards, 
and $H$ is the horizon length. At each time step $h$ and state $s_h$ the $\max$ player selects action $a_h\in \mathcal{A}$ and the 
$\min$ player selects action $b_h\in \mathcal{B}$ and both observe reward $r_h(s_h,a_h,b_h)$. Then, the next state $s_{h+1}$ is sampled from $p_h(\cdot | s_h,a_h,b_h)$. 

A strategy pair $(\pi, \nu)$ is comprised of the strategy of the first player $\pi = (\pi_1,\ldots,\pi_H)$, $\pi_h: \mathcal{S} \rightarrow \Delta (\mathcal{A})$ $\forall h \in [H]$ and that of the second player $\nu = (\nu_1,\ldots,\nu_H)$, $\nu_h:\mathcal{S}\rightarrow \Delta(\mathcal{B})$ $\forall h \in [H]$. Given $h\in [H]$, we define the state-value function and state-action value function as
\begin{align*}
    V^{\pi,\nu}_h(s_h) & = \mathbb{E}\left[ \sum^H_{t=h} r_h(s_t,a_t,b_t) | \pi,\nu, s_h\right] ~,\\
    Q^{\pi,\nu}_h(s_h,a_h,b_h) & = \mathbb{E}\left[ \sum^H_{t=h} r_h(s_t,a_t,b_t) | \pi,\nu, s_h, a_h, b_h \right] ~.
\end{align*}

\subsection{Nash Equilibria and Performance Metrics}
Let $\nu$ be a fixed strategy of the second player. Then, an optimal policy with respect to the MDP induced by $\nu$ is called the best response of the first player and we denote it by $\text{br}(\nu)$. Similarly, for a fixed strategy $\pi$ of the first player, the best response of the second player is denoted by $\text{br}(\pi)$.  Further, for any $\pi$, $\nu$ and $h \leq H$, we define
\begin{align*}
    V^{\pi,*}_h(s_h) & = V^{\pi,\text{br}(\pi)}_h(s_h)=\inf_\nu V^{\pi,\nu}_h(s_h)~, \\
    V^{*,\nu}_h(s_h) & = V^{\text{br}(\nu),\nu}_h(s_h)=\sup_\mu V^{\mu,\nu}_h(s_h)~.
\end{align*}
\begin{defn}\label{def:nash_equilibrium}
    A Nash Equilibrium (NE) is a strategy pair $(\pi^*,\nu^*)$ such that, for all $h\in[H]$ and $s\in\mathcal{S}$:
    \begin{align*}
        \sup_\pi \inf_\nu V^{\pi,\nu}_h(s) = V^{\pi^*,\nu^*}_h(s)=\inf_\nu\sup_\pi V^{\pi,\nu}_h(s)~.
    \end{align*}
\end{defn}
It is well-known that the NE of a zero-sum Markov game with a unique value function exists \citep{shapley1953stochastic}. We define $V^*_h(s_h)=V^{\pi^*,\nu^*}_h(s_h)$, for all $h \leq H$. Then, for all strategy pairs $(\pi,\nu)$, the weak duality is written as $V^{\pi,*}_h(s_h) \leq V^*_h(s_h) \leq V^{*,\nu}_h(s_h), \forall h \leq H$.
Consequently, the suboptimality gap of $(\pi,\nu)$ is 
\begin{align*}
    \optgap (\pi,\nu,s) = V^{*,\nu}_1(s) - V^{\pi,*}_1(s)~.
\end{align*}
Note that the duality is always non-negative, and it is zero only when $(\pi,\nu)$ is a NE strategy. We say that a strategy $(\pi,\nu)$ is a $\eta$-approximate NE if we have $\optgap (\pi,\nu, s) \leq \eta$ for all $s\in\mathcal{S}$.

\subsection{Linear Markov Games}\label{sec:linear_markov_games}

We consider linear two-player zero-sum Markov games $\mathcal{G}= (\mathcal{S},\mathcal{A},\mathcal{B},p,r,H)$. Here, we formally state the linearity assumption for Markov games, standard in the literature \citep{xie2020learning}.
\begin{defn}[Linear Markov games]\label{asmp_linearity}
 For each $(s,a,b)\in \mathcal{S}\times \mathcal{A}\times \mathcal{B}$ we have
\begin{align*}
    r_h(s,a,b) & = \phi(s,a,b)^\top \theta_h + \zeta \;\; \text{and} \\
    p_h(\cdot |s,a,b) & = \phi(s,a,b)^\top \mu_h(\cdot)~,
\end{align*}
where $\phi : \mathcal{S}\times \mathcal{A} \times \mathcal{B} \rightarrow \mathbb{R}^d$ is a known feature map, $\theta_h \in \mathbb{R}^d$ is an unknown vector, $\zeta$ is zero-mean $\gamma^2$-subGaussian noise, and $\mu_h=(\mu^{(i)}_h)_{i \in [d]}$ is a vector of $d$ unknown signed measures on $\mathcal{S}$. We assume that $\norm{\phi(\cdot,\cdot,\cdot)}_2 \leq 1$, $\norm{\theta_h}_2\leq \sqrt{d}$, and $\norm{\mu_h(\mathcal{S})}_2\leq\sqrt{d}$ for all $h\in [H]$. 
\end{defn}
As previously observed \citep{zhong2022pessimistic}, given a policy pair $(\pi,\nu)$ and time-step $h$, there exists a $d$-dimensional weight vector $\omega^{\pi,\nu}_h$, with $\norm{\omega^{\pi,\nu}_h}_2\leq H\sqrt{d}$, such that $Q^{\pi,\nu}_h(s,a,b) = \phi(s,a,b)^\top \omega^{\pi,\nu}_h$ for any $(s,a,b)$ tuple. 
 
\subsection{Offline Data Collection}

In offline RL, the objective is to learn an optimal policy from data that has already been collected beforehand \citep{levine2020offline}. Similarly, in offline Markov Games, the objective is to learn an approximate NE strategy based on a given dataset. Formally, we are given a dataset $ D = \{ (s^\tau_h,a^\tau_h,b^\tau_h,r^\tau_h,s^\tau_{h+1}) \}^{\tau\in [K]}_{h\in [H]}$ of $K$ trajectories,
gathered from a behavioral policy $\rho = (\rho_1,\ldots,\rho_H), \rho_h:\mathcal{S} \rightarrow \Delta(\mathcal{A}\times \mathcal{B})$, $\forall h \in [H]$.\footnote{The behavioral policy can be thought of as a product of a $\min$ and $\max$ policy.}

Given $h\in [H]$, a strategy pair $(\pi,\nu)$, and a tuple $(s,a,b)$, we denote by $d^{\pi,\nu}_h(s,a,b)$ the probability that the state-action tuple $(s,a,b)$ is traversed in time step $h$ by $(\pi,\nu)$, i.e., $d^{\pi,\nu}_h(s,a,b) = \mathbb{P}\left( s_h=s,a_h=a,b_h=b|\pi,\nu\right)$.
If $d^{\pi,\nu}_h(s,a,b) >0$, for all $h\in [H]$, we say that state-action tuple $(s,a,b)$ is covered by policy pair $(\pi,\nu)$. 

Next, we formally define the compliance of a given offline dataset with an underlying Markov game, which basically implies that the clean collected data follows the same dynamics as the environment. We will assume later on that the  clean dataset is in compliance with the underlying Markov game.
\begin{defn}[Compliance of dataset]  Given a Markov game $ \mathcal{G} = (\mathcal{S},\mathcal{A},\mathcal{B},p,r,H)$ and a dataset $D = \{ (s^\tau_h,a^\tau_h,b^\tau_h,r^\tau_h,s^\tau_{h+1}) \}^{\tau\in [K]}_{h\in [H]}~,$ we say the dataset $D$ is compliant with $\mathcal{G}$ if, for all $h\in[H]$ and $s\in\mathcal{S}$, 
\begin{align*}
    \mathbb{P} & \left( r^\tau_h =r,s^\tau_{h+1}=s \vert \{ (s^i_h,a^i_h,b^i_h)\}^\tau_{i=1}, \{ (r^i_h,s^i_{h+1})\}^{\tau -1}_{i=1} \right) \\
        & = \mathbb{P}_h\left( r_h=r,s_{h+1}=s \vert s_h=s^\tau_h, a_h=a^\tau_h,b_h=b^\tau_h \right)~,
\end{align*}
where $\mathbb{P}$ is with respect to $D$ and $\mathbb{P}_h$ is the probability measure taken with respect to the underlying Markov game $\mathcal{G}$.
\end{defn}

\subsection{Corruption Robust Estimation}

The standard assumption in statistical estimation is that the samples we are given come from a fixed distribution, allowing one to directly use probability laws to obtain unbiased estimates of interest. However, it is usually the case that outliers are present in the data that do not belong to the underlying distribution, or that an adversary can arbitrarily corrupt the data. Only one outlier is enough to arbitrarily shift the empirical mean of the data, and therefore acquiring robust estimators for the moments of the distribution is important. 
We will use the Huber contamination model, akin to the corruption model in single-player offline RL \citep{zhang2022corruption}.
\begin{assumption}[$\epsilon$-contamination in offline Markov games]\label{asm:corruption_model}
    Given $\epsilon \in [0,1]$ and a set of clean tuples $\widetilde{D}= \{ (\widetilde{s}_i,\widetilde{a}_i,\widetilde{b}_i,\widetilde{r}_i,\widetilde{s}'_i)\}^N_{i=1}$, an adversary is allowed to inspect the tuples and replace any $\epsilon N$ of them with arbitrary contaminated tuples $(s,a,b,r,s')\in \mathcal{S}\times\mathcal{A}\times\mathcal{B}\times\mathcal{R}\times\mathcal{S}$. The resulting set $D=\{ (s_i,a_i,b_i,r_i,s'_i)\}^N_{i=1}$ is then revealed to the learner. 
\end{assumption}

We say that a set of samples is $\epsilon$-corrupted if it is generated by the above process. Given an $\epsilon$-corrupted set of data points, the goal of robust statistics is to compute accurate estimates of the first and second moments. In \citep{diakonikolas2017being}, the authors provide efficient and nearly sample-optimal filtering algorithms for mean and covariance estimation. 

For the linear regression problem, different robust estimators require different coverage assumptions. For instance, SCRAM \citep{chen2022online} does not make strong coverage assumptions but assumes clean covariates, while Robust Least Squares (RLS) \citep{bakshi2021robust, pensia2020robust, zhang2022corruption} needs stronger coverage while allowing for corrupted covariates. We will use both aforementioned methods as robust estimator oracles under different coverage assumptions and corruption models.

\section{PROBLEM FORMULATION}\label{sec:problem_formulation}

We assume that an unknown experimenter collects a dataset $\widetilde{D}=$ \cleandata of $K$ trajectories, in compliance with $\mathcal{G}$. We assume that \cleandata is collected from behavioral policy 
$\rho=(\rho^1,\rho^2)$. Formally, for any $h\in[H]$, we have $(\widetilde{s}_h,\widetilde{a}_h,\widetilde{b}_h) \sim d^\rho_h$ and $\widetilde{s}_{h+1}\sim p_h(\cdot|\widetilde{s}_h,\widetilde{a}_h,\widetilde{b}_h)$ for any $(\widetilde{s}_h,\widetilde{a}_h,\widetilde{b}_h,\widetilde{r}_h,\widetilde{s}_{h+1})\in \widetilde{D}$.  Subsequently, an adversary contaminates an $\epsilon$-fraction of all tuples of $\widetilde{D}$ and provides us with the corrupted dataset $D=$\data. 

The learner's objective is, given access to an $\epsilon$-corrupted dataset $D$, to be able to compute an approximate Nash equilibrium policy pair. In the clean offline setting, this is usually done by leveraging pessimism in order to penalize samples that were not sufficiently covered. However, if the dataset does not cover state actions of interest, i.e., those traversed by equilibrium policies, then learning becomes impossible, even in the clean setting. First, we formally define what we mean by coverage.
\begin{defn}
    A strategy pair $(\pi, \nu)$ is covered by the dataset generated by behavioral policy $\rho$ if and only if every tuple $(s,a,b)$ that is covered by $(\pi, \nu)$ is also covered by $\rho$. In other words, we have
\begin{align*}
    \frac{d^{\pi, \nu}_h(s,a,b)}{d^{\rho}_h(s,a,b)} < \infty, \forall (s,a,b) \in \mathcal{S}\times\mathcal{A}\times\mathcal{B}, h \in [H]~.
\end{align*}
\end{defn}
For the offline single-player setting, it has been established that coverage of the optimal (or any target policy used as reference) policy is necessary to achieve convergence. For the two-player zero-sum Markov game setting, \cite{zhong2022pessimistic} show that coverage of a Nash equilibrium policy pair and its neighbors across each player is necessary for learning, as we will see in the next section.\footnote{Formal definitions are given in the next section.} We will consider the corruption problem under different coverage assumptions, starting from strong assumptions on the corrupted data, and ending with the setting where minimal coverage assumptions hold only on the clean dataset.

\section{RESULTS UNDER COVERAGE ON CORRUPTED DATA}\label{sec:linear_setting}

Recall from Section \ref{sec:linear_markov_games} that, given $(\pi, \nu)$, $(s,a,b)\in \mathcal{S}\times \mathcal{A} \times \mathcal{B}$, and $h\in [H]$, Definition \ref{asmp_linearity} implies that $Q^{\pi,\nu}_h(s,a,b) =\phi(s,a,b)^\top \omega^{\pi,\nu}_h$, for some $\omega^{\pi,\nu}_h\in\mathbb{R}^d$.
Thus, learning the action value function $Q^*_h$ that corresponds to a Nash equilibrium (NE) pair reduces to learning the optimal weights, which we denote by $\omega^*_h$. For that, we will rely on \textit{Pessimistic Minimax Value Iteration} (PMVI) \citep{zhong2022pessimistic}, which computes an approximate NE pair using offline data.
However, PMVI computes estimates of $\omega^*_h$ by solving regularized least-squares on the Bellman operator. We cannot directly use those estimates since our data contains corrupted samples. Instead, we use a robust estimator to tackle the corruption problem. Let \textrm{R-EST} denote a generic robust estimator. 

First, we randomly split the data into $H$ batches $D_h$  and set $\underline{V}_{H+1}(\cdot)=\overline{V}_{H+1}(\cdot)=0$. Then, for each time step $h=H, H-1,\ldots,1$, we obtain
\begin{align*}
    \underline{\omega}_h & \leftarrow \textrm{R-EST} \Bigl( \{ \phi(s^\tau_h,a^\tau_h,b^\tau_h), r^\tau_h + \underline{V}_{h+1}(s^\tau_{h+1})\}^K_{\tau=1} \Bigr)~, \\
    \overline{\omega}_h & \leftarrow \textrm{R-EST} \Bigl( \{ \phi(s^\tau_h,a^\tau_h,b^\tau_h), r^\tau_h + \overline{V}_{h+1}(s^\tau_{h+1})\}^K_{\tau=1} \Bigr)~.
\end{align*}
Once we have estimates of the optimal weights, the algorithm proceeds to construct estimates of the $Q$-values for both players, which depend on carefully constructed bonus terms $\Gamma_h$ -- we provide different bonus terms depending on the robust estimators we use and other characteristics of corruption. 

Next, we compute a Nash equilibrium corresponding to payoffs $\underline{Q}(\cdot,\cdot,\cdot) $ and $\overline{Q}(\cdot,\cdot,\cdot) $, and obtain $(\widehat{\pi}_h,\nu'_h)$ and $(\pi',\widehat{\nu}_h)$, as solutions corresponding to $\underline{Q}(\cdot,\cdot,\cdot)$ and $\overline{Q}(\cdot,\cdot,\cdot)$, respectively. 

Finally, the algorithm estimates the value functions for both players based on the computed policy pairs. After $H$ steps, the algorithm terminates and outputs the estimated pairs of strategies $(\widehat{\pi},\widehat{\nu})$. The pseudocode of the described method is given in Algorithm \ref{alg:generic_PMVI}.

\begin{algorithm}
\setstretch{1.2}
\caption{Robust PMVI}
\label{alg:generic_PMVI}
\begin{algorithmic}[1]
\State \textbf{Input}: Dataset $D$, failure probability $\delta$, robust estimator \textrm{R-EST}, bonus functions $\Gamma_h(\cdot,\cdot,\cdot)$.
\State \textbf{Initialize:} Randomly split dataset $D$ into $H$ subsets $D_h$ of cardinality $K$; set $\underline{V}_{H+1}(s)=\overline{V}_{H+1}(s)=0$, for all $s\in\mathcal{S}$.
\For{$h=H,H-1,\ldots,1$}:
\State Compute estimates $\underline{\omega}_h$ and $\overline{\omega}_h$ via \textrm{R-EST}.
\State  $\underline{Q}_h(\cdot,\cdot,\cdot) \leftarrow \Pi_{H-h+1}\left( \phi(\cdot,\cdot,\cdot)^\top \underline{\omega}_h - \Gamma_h(\cdot,\cdot,\cdot) \right)$. 
\State $\overline{Q}_h(\cdot,\cdot,\cdot)\leftarrow \Pi_{H-h+1}\left( \phi(\cdot,\cdot,\cdot)^\top \underline{\omega}_h + \Gamma_h(\cdot,\cdot,\cdot) \right)$.
\State Compute $(\widehat{\pi}_h, \nu'_h)$ and $(\pi'_h,\widehat{\nu}_h)$ as NE solutions to $\underline{Q}_h(\cdot, \cdot, \cdot)$ and $\overline{Q}_h(\cdot, \cdot, \cdot)$, respectively. 
\State  $\underline{V}_h(\cdot) \leftarrow \mathbb{E}_{a\sim \widehat{\pi}_h,b\sim \nu'_h}[\underline{Q}_h(\cdot,a,b)]$. 
\State $\overline{V}_h(\cdot) \leftarrow \mathbb{E}_{a\sim \pi'_h,b\sim \widehat{\nu}_h}[\overline{Q}_h(\cdot,a,b)]$.
\EndFor
\State \textbf{Output:} $\widehat{\pi} = (\widehat{\pi}_h)^H_{h=1},\widehat{\nu}=(\widehat{\nu}_h)^H_{h=1}$. 
\end{algorithmic}
\end{algorithm}

In the following sections, we will introduce different estimators that compute the Bellman operator weights under different coverage assumptions. Depending on which estimator we use, we also get the corresponding convergence guarantees. Intuitively, the error of estimating the Bellman operator is inflated by an extra term coming from data corruption. To have a sense of the strength of such guarantees in terms of $\epsilon$, we now present our first contribution, a result that states an algorithm-independent minimax lower bound for the corruption problem in the offline two-player zero-sum setting. 

\begin{thm}\label{thm:lower_bound}
   For every algorithm $L$, there exists a Markov game $\mathcal{G}$, an instance of the corrupted dataset, corruption level $\epsilon$, and a data collecting distribution $\rho$, such that, with probability at least $1/4$, $L$ will find a no-better than $\Omega(Hd\epsilon)$-approximate NE policy pair $(\widetilde{\pi},\widetilde{\nu})$. That is, with a probability of at least $1/4$, we have, for every $s\in\mathcal{S}$:
    \begin{align*}
        \optgap(\widetilde{\pi},\widetilde{\nu},s) = \Omega(Hd\epsilon)~.
    \end{align*}
\end{thm}

\subsection{Uniform $\Sigma$-Coverage and Corrupted Covariates}\label{sec:unifrom_corrupted}

We start by considering the strongest coverage assumption first. 
Throughout Section \ref{sec:unifrom_corrupted}, we assume that the corrupted data $D$ satisfies the following.
\begin{assumption}[Uniform $\Sigma$-coverage]\label{asn:uniform-phi-coverage}
 For all $h\in [H]$ and $s\in\mathcal{S}$, $\E_{\rho}\left[ \Sigma_h \vert s_1=s \right] \succeq \kappa I$ for some $\kappa > 0$, where $\Sigma_h=\phi(s_h,a_h,b_h) \phi(s_h,a_h,b_h)^\top$.
 \end{assumption}
 Assumption \ref{asn:uniform-phi-coverage} implies that every state-action tuple in the support of the behavioral policy $\rho$ is covered by the offline data. In Section \ref{sec:coverage_discussion} we further discuss the strength of this assumption relative to other coverage assumptions studied in the literature.

\looseness-1Inspired by the Robust Value Iteration algorithm \citep{zhang2022corruption} that uses a Robust Least Squares (RLS) oracle as a subroutine for estimating the weights of the value function, we propose RLS-PMVI, which uses a similar oracle that relies on Assumption \ref{asn:uniform-phi-coverage}. However, it allows for an arbitrary corruption model, where an $\epsilon$-fraction of the dataset can be arbitrarily corrupted, that is, any component of the data samples can be modified arbitrarily. RLS-PMVI returns an approximate NE policy pair that incurs the following bounds on the gap. 

\begin{thm}\label{thm:linear_uniform_coverage}
    Suppose that Assumption \ref{asn:uniform-phi-coverage} holds on an $\epsilon$-corrupted dataset $D$ corresponding to a linear Markov game. Then, given $\delta > 0$, with probability at least $1-\delta$, RLS-PMVI with bonus term $\Gamma_h(s,a,b)=0$ achieves suboptimality gap upper bounded by
    \begin{align*}
         \widetilde{O} \left( \sqrt{\frac{H(H+\gamma)^2poly(d)}{\kappa^2 K}} + \frac{H(H+\gamma)}{\kappa}\epsilon \right)~.
    \end{align*}
\end{thm}
The proof of Theorem \ref{thm:linear_uniform_coverage} is based on similar ideas to those used in the single-player setting. We provide the full proof in Appendix for completion. The following remarks are in order.
\begin{rem}
    Note that the order of $\epsilon$ is optimal. Furthermore, since the bonus term is $0$, the algorithm \textit{does not require knowledge of $\epsilon$.} Thus, if we have uniform coverage, under a stricter corruption model, i.e., the features can also be corrupted, and without knowledge of $\epsilon$, RLS-PMVI  incurs a suboptimality gap with optimal dependence on $\epsilon$.
\end{rem}
\begin{rem}
    Note that the coverage constant $\kappa$ is in the order of $1/d$, since $\norm{\phi}_2 \leq 1$. Thus, the bound corresponding to the corruption error becomes $O(H^2d\epsilon)$. 
\end{rem}

\subsection{LRU Coverage and Clean Covariates}\label{sec:lru_corrupted}

As mentioned in the previous section, Assumption \ref{asn:uniform-phi-coverage} is a very strong requirement to make on the data. Thus, we now focus our attention on a relaxed scenario, where the observed data covers only policies of interest. In the two-player zero-sum setting, it turns out that coverage of the NE policy pair alone is not enough to efficiently compute an approximate solution. We state such an assumption below. Given dataset $\mathcal{D}$ of $K$ trajectories, let us denote by 
\begin{align}\label{eq:sample_covariance}
    \Lambda_h = \sum^K_{\tau=1}\phi(s^\tau_h,a^\tau_h,b^\tau_h)\phi (s^\tau_h,a^\tau_h,b^\tau_h)^\top + I
\end{align}
the regularized sample covariance with respect to $\mathcal{D}$.

 \begin{assumption}[Low relative uncertainty]\label{asm:low_relative_uncertainty}
    There exists a constant $c_1>0$ such that, for all $x\in \mathcal{S}$:
    \begin{align*}
        \Lambda_h  \succeq I + c_1 K & \max \left\{ \sup_\nu \; \mathbb{E}_{\pi^*,\nu}[\phi_h\phi_h^\top \vert s_1=x], \right. \\ & \quad \quad \left. \sup_\pi \; \mathbb{E}_{\pi,\nu^*}[\phi_h\phi_h^\top \vert s_1=x] \right\}~,
    \end{align*}
    where $\phi_h = \phi(s_h,a_h,b_h)$, for any $h\in[H]$.
\end{assumption}

As shown in \citep{zhong2022pessimistic}  (and \citep{cui2022offline} for tabular settings), Assumption \ref{asm:low_relative_uncertainty} (and Assumption \ref{asn:unilateral-coverage} for tabular settings\footnote{See Section \ref{sec:coverage_discussion}.}) is necessary for learning NE policies in two-player settings. Thus, the focus of Section \ref{sec:lru_corrupted} will be on data that satisfies such an assumption.
An immediate problem that naturally follows is choosing the right robust estimator. The RLS estimator, described in the previous section, provides nice guarantees but it relies on Assumption \ref{asn:uniform-phi-coverage}. If we are to assume only LRU coverage on the data, then we need a different estimator. 

To that end, we utilize a result by \cite{chen2022online} that does not require Assumption \ref{asn:uniform-phi-coverage} to hold. Their algorithm, SCRAM, utilizes an alternating minimization scheme to compute first-order stationary points.  However, the utilization of such an oracle relies on the assumption that the covariates of the dataset are not corrupted. Translated into our setting, this assumption requires the features $\phi(s,a,b)$, and, as a consequence, tuples $(s,a,b)\in D$ to be clean, i.e., only an $\epsilon$-fraction of the rewards and the next states are allowed to be arbitrarily corrupted. Below, we state the conditions of SCRAM and its guarantee. A  detailed version can be found in  Appendix.
\begin{lem}[\cite{chen2022online}]\label{asmp_scram}
Given a dataset $D=\{ x_i, y_i\}_{i \in [K]}$, which is an $\epsilon$-corrupted version of dataset $\widetilde{D}= \{ x_i,\widetilde{y}_i\}$, where $\widetilde{y}_i=\langle \omega^*_i,x_i\rangle + \xi_i$, $\epsilon < 0.499$, $\norm{x_i}_2\leq 1$ and $\xi_i$ are conditionally zero-mean $\gamma^2$ sub-Gaussian, then SCRAM returns estimators $(\omega_k)_{k\in[K]}$, such that
\begin{align}\label{eq:scram_approximate}
    \norm{\omega - \omega^*}_{\Sigma} \leq O(\epsilon\gamma + \gamma d^{1/2}K^{-1/4})~,
\end{align}
omitting poly-log factors, where $\Sigma = (1/K)\sum^K_{k=1}x_k x^\top_k$ is the sample covariance.\footnote{For explicit bounds, see Appendix.}
\end{lem}

In Section \ref{sec:lru_corrupted}, we use \textrm{SCRAM} as a robust estimator with bonus term defined as
\begin{align*}
    \Gamma_h(s,a,b) =\left(\sqrt{K} \mathcal{E} + 2H\sqrt{d} \right) \norm{\phi(s,a,b)}_{\Lambda_h^{-1}}~,
\end{align*}
where $\Lambda_h$ is defined in Equation \ref{eq:sample_covariance} and $\mathcal{E}$ denotes the upper bound in Equation \eqref{eq:scram_approximate}.

We are now ready to state the main result of Section \ref{sec:lru_corrupted}, which gives an upper bound on the suboptimality gap of the policy pair returned by SCRAM-PMVI. 

\begin{thm}\label{thm:linear_case}
    Suppose that Assumption \ref{asm:low_relative_uncertainty} holds with given constant $c_1$. Let $\delta >0$, $\epsilon < 1/2$ and let $D$ be the $\epsilon$-corrupted version of the dataset $\widetilde{D}$ comprised of $K$ trajectories of length $H$, where $K \geq \log (\min (K,d))/\epsilon$ and $(\widetilde{s}^\tau_h,\widetilde{a}^\tau_h,\widetilde{b}^\tau_h) = (s^\tau_h,a^\tau_h,b^\tau_h)$, for all $\tau \in [K]$, $h\in[H]$. Then, with probability at least $1-\delta$, SCRAM-PMVI outputs $(\widehat{\pi},\widehat{\nu})$ that satisfy, for every $s\in\mathcal{S}$:
    \begin{align*}
        \optgap (\widehat{\pi},\widehat{\nu},s) & \leq \widetilde{O}\left( \frac{1}{\sqrt{c_1}}(\gamma+H)H\sqrt{d}\epsilon + \frac{H^2d}{\sqrt{c_1K}}\right)~.
    \end{align*}
\end{thm}

    \looseness-1Note the $\sqrt{d}$ factor in the term coming from the corruption error. This might seem contradictory to our lower bound of Theorem \ref{thm:lower_bound} at first. However, there is a hidden dependence on $d$ in the $c_1$ constant, since $c_1$ cannot be arbitrarily large. By Assumption \ref{asm:low_relative_uncertainty}, first one can easily see that $0<c_1\leq 1$. Moreover, as shown in the proof of Proposition \ref{prop:equivalent_assumptions} (see Appendix), $c_1$ can be $O(1/d)$, in which case the bounds would be written as 
         $$\optgap (\widehat{\pi},\widehat{\nu},s)  \leq \widetilde{O}\left( (\gamma+H)Hd\epsilon + H^2d^{3/2}K^{-1/2}\right).$$
    
    Note that one can always construct other examples in which the dependence of the gap on $d$ is looser. However, here we are interested in showing that there is an explicit dependence in order to confirm the matching order.

It is worth mentioning that knowledge of $\epsilon$ is required by both SCRAM and RLS oracles. Moreover, agnostic learning without knowledge of $\epsilon$ is  impossible for linear Markov games, without uniform data coverage, as the next result shows. The proof follows immediately from that of Theorem 3.4 of \citep{zhang2022corruption}, by restricting the second player's action space to a single action and observing that the LRU coverage reduces to their finite relative condition number assumption.

\begin{cor}\label{thm:agnostic_learning}
    Under Assumption \ref{asm:low_relative_uncertainty}, for every algorithm $L$ that achieves a diminishing suboptimality gap in a clean environment, there exist linear Markov games $\mathcal{G}_1$ and $\mathcal{G}_2$, an instance of corrupted data, and a data collecting distribution $\rho$, such that, for every $\epsilon \in (0,1/2]$, $L$ achieves $\optgap (\widehat{\pi},\widehat{\nu},s) \geq 1/2$, for any $s\in\mathcal{S}$, with probability at least $1/4$, on at least one of the games.
\end{cor}

With this result, we end our discussion under corrupted data coverage assumptions. In order to give a formal characterization of the used coverage assumptions and other similar ones in the literature, we provide a formal discussion in Section \ref{sec:coverage_discussion}. In the next section, we turn our attention to the more difficult setting where coverage guarantees are no longer present in the corrupted data and provide near-optimal bounds on the suboptimality of SCRAM-PMVI with novel bonus terms designed to account for this corrupted coverage. 

\section{RESULTS UNDER COVERAGE ON CLEAN DATA}\label{sec:no_coverage_bounds}

In previous sections, we considered the problem of corruption in the linear setting, under the assumption that the corrupted data preserves the coverage that is necessary to solve the underlying Markov game. However, this might be too restrictive of an assumption since, in the worst case, the attacker would corrupt precisely those state-action tuples that belong to trajectories covered by LRU policies. 
Thus, the next natural question is whether we can approximately solve the problem given that we are not guaranteed that the available corrupted data satisfies the LRU coverage assumption. 

Apart from not assuming any guarantees on coverage, we also assume fully arbitrary corruption, in the sense that the attacker can arbitrarily corrupt an $\epsilon$-fraction of both covariates and observations. If we are to apply \textrm{SCRAM} in this setting, we would expect an additional error coming from the corruption of covariates, and another one coming from the corrupted coverage. 

\looseness-1We find that by carefully designing a bonus term that also takes into account these two additional errors, we are able to approximately solve the corruption robust offline two-player zero-sum game, as the following result states.

\begin{thm}\label{thm:general_result}
    Suppose that the condition of Assumption \ref{asm:low_relative_uncertainty} is satisfied only on the clean dataset $\widetilde{D}$, for a given constant $c_1$, that is, assume that the clean sample covariance matrix $\sum^{K}_{\tau =1} \widetilde{\phi}^\tau_h (\widetilde{\phi}^\tau_h)^\top + I,$
    where $\widetilde{\phi}^\tau_h$ denotes the clean feature of the sample at time-step $h$ of episode $\tau$, satisfies Assumption \ref{asm:low_relative_uncertainty}.
    Furthermore, let $\delta > 0$, $\epsilon \in (0,1/2)$, and $K \geq \log (\min \{ K,d\}) /\epsilon$. Then, under Assumption \ref{asm:corruption_model}, with probability at least $1-\delta$, SCRAM-PMVI with bonus defined as
    \begin{align*}
       \Gamma_h(\cdot) = \left(\sqrt{(1-\epsilon)K}\mathcal{E} + (\sqrt{\epsilon K} + 2) H\sqrt{d} \right) \norm{\phi(\cdot)}_{\Lambda^{-1}_h}
    \end{align*}
    returns $(\widehat{\pi},\widehat{\nu})$ that satisfy, for every $s\in\mathcal{S}$:
    \begin{align*}
        \optgap(\widehat{\pi},\widehat{\nu},s) \leq   O\left( \frac{1}{\sqrt{c_1}}(\gamma+H)Hd\sqrt{\epsilon} + \frac{H^2d}{\sqrt{c_1K}}\right)~.
    \end{align*}
\end{thm}
Note that we incur an additional $\sqrt{d}$ factor in the corruption error and that the order of $\epsilon$ is not optimal. This is due to both the error from corrupted covariates and the one from the corrupted coverage. Thus, the next natural question is whether these bounds can be improved, at least in terms of the dependence on $\epsilon$. We answer this question in the affirmative. We show that, under a mild assumption on the feature space, we can improve the order of $\epsilon$ at the cost of an extra $d^{3/2}$ factor, thus recovering the optimal dependency on the corruption level. First, we state the assumption. 

\begin{assumption}\label{asm:features_lower_bound}
    Given a linear Markov game as in Definition \ref{asmp_linearity}, we assume that the features satisfy $\min_{s,a,b}\norm{\phi(s,a,b)}_2 \geq c_2$, for some $c_2>0$. 
\end{assumption}

We emphasize that such an assumption is not restrictive and that $c_2$ is in the order of $1/\sqrt{d}$ under various feature constructions, such as random Fourier features  \citep{rahimi2007random}. We find that Assumption \ref{asm:features_lower_bound} is enough to obtain order-optimal bounds in terms of $\epsilon$, albeit having an additional $O(d)$ term when comparing with Theorem \ref{thm:general_result} coming from the constant $c_2$.

\begin{thm}\label{thm:second_general_result}
    Suppose that the conditions of Theorem \ref{thm:general_result} and Assumption \ref{asm:features_lower_bound} hold. Then, with probability at least $1-\delta$, SCRAM-PMVI with bonus $\Gamma_h(\cdot)$ defined as
    \begin{align*}
          \left( 2(1-\epsilon)K\mathcal{E} + \epsilon KH\sqrt{d} + H\sqrt{Kd}\right) \norm{\phi(\cdot)^\top\Lambda^{-1}_h}_2~,
    \end{align*}
    returns $(\widehat{\pi},\widehat{\nu})$ that satisfy, for every $s\in\mathcal{S}$:
    \begin{align*}
        \optgap(\widehat{\pi},\widehat{\nu},s) \leq \widetilde{O} \left( \frac{1}{c_1c_2}H^2d^{3/2}\epsilon + \frac{H^2d^{3/2}}{c_1c_2\sqrt{K}} \right)~.
    \end{align*}
\end{thm}\vspace{-0.6em}

This improvement comes as a result of a novel bonus term and a different style of analysis based on an application of the Woodbury matrix identity to account for the extra terms coming from corruption.

\section{DISCUSSION ON MARL COVERAGE ASSUMPTIONS}\label{sec:coverage_discussion}

In this section, we discuss the relationship between various coverage assumptions used in the literature.
First, we state three additional assumptions, apart from Assumption \ref{asn:uniform-phi-coverage} and \ref{asm:low_relative_uncertainty}.
 
 \begin{assumption}[Single-policy coverage]\label{asn:single-coverage} The NE strategy pair $(\pi^*,\nu^*)$ is covered by the dataset $D$.
 \end{assumption}
 
 \begin{assumption}[Unilateral coverage]\label{asn:unilateral-coverage}
 For all strategies $\pi:\mathcal{S}\rightarrow \Delta(\mathcal{A})$ and $\nu: \mathcal{S}\rightarrow\Delta(\mathcal{B})$, the strategy pairs $(\pi^*,\nu)$ and $(\pi,\nu^*)$ are covered by $D$.
 \end{assumption}
 
 \begin{assumption}[Uniform coverage]\label{asn:uniform-coverage}
 For all $h\in [H]$ and $(s,a,b)\in \mathcal{S}\times\mathcal{A}\times\mathcal{B}$, the tuple $(s,a,b)$ at time step $h$ is covered by $D$.
 \end{assumption}


Note that Assumption \ref{asn:single-coverage} is the weakest assumption. Moreover, it is a direct extension of Assumption 3.2 of \citep{zhang2022corruption}, while Assumption \ref{asn:uniform-coverage} is an extension of the uniform policy coverage in the single-player setting \citep{cai2020provably}. We show that Assumption \ref{asn:unilateral-coverage} implies Assumption \ref{asm:low_relative_uncertainty} with high probability, when the feature matrix has full rank, and that, for tabular Markov games,\footnote{For tabular Markov games, we have $\phi(s,a,b)=e_{s,a,b}$, where $e_{s,a,b}$ is the $SAB$-dimensional zero vector with $1$ in the $(s,a,b)$ entry.} these two assumptions are equivalent. All proofs of our results can be found in the Appendix.

\looseness-1\begin{pro}\label{prop:equivalent_assumptions}
   Let $\Phi \in \mathbb{R}^{SAB\times d}$ denote the feature matrix. Assume $\Phi$ has full rank and let $\delta\in(0,1)$. Then, if Assumption \ref{asn:unilateral-coverage} holds, there exists a positive constant that depends on $\delta$ for which  Assumption \ref{asm:low_relative_uncertainty} holds, with probability at least $1-\delta$. Moreover, in the tabular Markov game setting, these two assumptions are equivalent. 
\end{pro}

Furthermore, it is obvious that uniform coverage is stronger than unilateral coverage. The relation between Assumption \ref{asn:uniform-phi-coverage} and Assumption \ref{asn:uniform-coverage} is given as follows.
\looseness-1\begin{pro}\label{prop:strength_of_assumptions}
    Assume that $\Phi$ has full rank. Then, if Assumption \ref{asn:uniform-phi-coverage} holds, Assumption \ref{asn:uniform-coverage} holds. Moreover, in the tabular MG setting, these two assumptions are equivalent. 
\end{pro} 

 As already shown in \citep{zhong2022pessimistic}, assuming that the collected data covers only the NE strategy pair $\pi^*=(\mu^*,\nu^*)$ \text{is not enough} to learn an approximate NE policy pair. Indeed, Assumption \ref{asm:low_relative_uncertainty} is necessary for solving offline two-player zero-sum Markov games, even in clean environments.

 \begin{rem}\label{rem:coverage_assumptions}
     Note that, under the assumption that $\Phi$ is full rank (if not,  orthogonalization can be applied), the given coverage assumptions are listed according to their strength. With high probability, we have
     \begin{equation*}
         A. \ref{asn:uniform-phi-coverage} \Rightarrow A. \ref{asn:uniform-coverage} \Rightarrow A. \ref{asn:unilateral-coverage} \Rightarrow A. \ref{asm:low_relative_uncertainty} \Rightarrow A. \ref{asn:single-coverage}~.
     \end{equation*}
     Moreover, for the tabular setting, Assumption \ref{asn:uniform-phi-coverage} is equivalent to \ref{asn:uniform-coverage} and Assumption \ref{asn:unilateral-coverage} is equivalent to \ref{asm:low_relative_uncertainty}. 
 \end{rem}


\section{RELATED WORK}

\textbf{Offline RL.} Our work is related to the offline RL literature, where there have been substantial developments in recent years, both on the empirical front \citep{jaques2019way, laroche2019safe, fujimoto2019off, kumar2020conservative, agarwal2020optimistic, kidambi2020morel} and the theoretical front \citep{jin2021pessimism, xie2021bellman, rashidinejad2021bridging, uehara2021pessimistic, zanette2021provable}. As previously mentioned, coverage assumptions on the data are key in this setting, and there has been a variety of different assumptions in the single-player setting, starting from the all-policy coverage $\norm{d^\pi/d^\rho}_\infty \leq B$, for all $\pi$ \citep{munos2008finite}, to optimal policy coverage $\norm{d^*/d^\rho}_\infty \leq B$ \citep{xie2021bellman}, and $\alpha$-regularized optimal policy coverage $\norm{d^*_\alpha/d^\rho}_\infty \leq B$ \citep{zhan2022offline}. First, all of these assumptions are with respect to the original data-collecting distribution $d^\rho$, while our weakest assumption (LRU coverage) relies only upon the sample covariance matrix of the data. Second, note that the weakest of the assumptions above ($\alpha$-regularized optimal policy coverage) requires coverage of only  the optimal solution to the regularized LP problem, while coverage of the NE policy pair and its unilateral neighbors is necessary for finding NE policy pairs in our setting \citep{zhong2022pessimistic}. This is arguably due to the higher complexity of the problem compared to the single-player setting. 

\textbf{Adversarial attacks in RL.} Our work also adds to the vast literature on adversarial attacks in ML \citep{szegedy2013intriguing, biggio2013evasion, nguyen2015deep, papernot2017practical, biggio2012poisoning, li2016data, xiao2012adversarial} and the existing body of work on adversarial attacks in RL and MARL \citep{huang2017adversarial, lin2017tactics, wu2022reward, gleave2019adversarial, sun2020stealthy, sun2020vulnerability, ma2019policy, rakhsha2021policy, everitt2017reinforcement, huang2019deceptive, rangi2022understanding, mohammadi2023implicit}. Specifically, we consider the problem of robustness to data corruption, which is a type of training-time attack \citep{mei2015using, xiao2015feature, rakhsha2020policy}. Popular types of defense against such attacks include randomized smoothing \citep{cohen2019certified, wu2021crop}, outlier detection \citep{diakonikolas2019sever} and robust estimator methods \citep{chen2022online, diakonikolas2017being, banihashem2023defense}. In this work, we use the latter methods, for both weight estimation in linear games and mean estimation in tabular ones. Arguably, the closest work to ours is \citep{zhang2022corruption}, which studies the corruption problem in single-agent RL. On the other hand, while our analysis of the lower bounds is inspired by their work, our analysis leads to tighter upper bounds in terms of $\epsilon$ and $H$ for two-player zero-sum Markov games. \cite{yang2022rorl} also study the robustness problem in the offline RL setting. However, their attack model assumes observation perturbations of bounded radius. Our attack model is stronger since it allows for the arbitrary perturbation of the tuples.  Recently, adversarial corruption in the online setting has been studied in linear contextual bandits \citep{he2022nearly} and more generally in MDPs with general function approximation \citep{ye2023corruption}. These works also broadly relate to corruption-robust approaches in distributed RL \citep{chen2022byzantine, fan2021fault}, that focus on MDP settings.

\looseness-1\textbf{Reward perturbation in MARL.} \citet{ma2022game} study the problem of reward design in online systems with no-regret learners. Their goal is to modify the utility function iteratively so that the agents converge to a desired action profile, while maintaining low cost of perturbation. While we also consider multi-agent systems subject to a third-party intervention, our focus is in the offline setting with an adversarial corruption framework, with an emphasis on the defense front. On the other hand, \citet{wu2023reward} is more closely related to ours. They also study an offline corruption model, where the attacker can perturb the reward signal in order to enforce a particular policy, while at the same time maintaining a low perturbation cost. While their focus in on designing cost-efficient attacks of that form, our work studies defenses against a broad class of data poisoning attacks, defined by the Huber contamination model.


\textbf{Learning in Markov games.}  Finally, our work also relates to the research area of learning in Markov games \citep{vrancx2008decentralized,  littman1994markov, littman2001value, tian2021online, wang2002reinforcement, sayin2021decentralized, xie2020learning}. In particular, we consider offline two-player zero-sum games, which have been recently considered in \citep{cui2022offline} for tabular settings and \citep{zhong2022pessimistic} for linear settings. While they solve the offline problem by assuming that the collected data is sampled from a benign behavioral policy, we consider the problem of robustness of their proposed methods to data corruption. More specifically, we assume that an $\epsilon$ fraction of collected data has been corrupted, and our aim is to learn NE strategy pairs under such an assumption. When $\epsilon=0$, we recover the same bounds as theirs on the suboptimality gap.

\section{CONCLUSION}\label{sec:discussion}


We considered the problem of data corruption in offline two-player zero-sum Markov games. Our contribution was to provide an extensive characterization of the problem under various coverage assumptions on both the clean and corrupted data. To the best of our knowledge, we are the first to provide such a characterization for the problem of corruption in offline Markov games. For the hardest setting where minimal coverage is guaranteed only on the clean data, we are able to match the optimal order of $\epsilon$ under mild structural assumptions, thus providing a full picture of this setting. There are many interesting future directions to pursue: i) studying robustness under adversarial corruption in Markov games with general function approximation; ii) extending the two-player zero-sum Markov game setting to online data corruption, where, in each round the reward/transition data is corrupted with probability $\epsilon$; iii) studying robustness to adversarial corruption in larger structured games (e.g. Markov potential games). 



\section*{Acknowledgements}

The authors thank anonymous reviewers for their valuable suggestions and comments. This research was, in part, funded by the Deutsche Forschungsgemeinschaft (DFG, German Research Foundation) – project number 467367360.

\bibliographystyle{plainnat}
\bibliography{bibliography}

\newpage
\onecolumn
\appendix

\addcontentsline{toc}{section}{} 
\part{Appendix} 
\parttoc 

\section{Proofs of Section \ref{sec:linear_setting}}\label{sec:missing_proofs}

In this section, we derive the proofs of the results in Section \ref{sec:linear_setting}.

\subsection{Proof of Theorem \ref{thm:lower_bound}}\label{sec:proof_lower_bound}

We first restate the result.

\begin{statement}
    For every algorithm $L$, there exists a Markov game $\mathcal{G}$, an instance of the corrupted dataset, corruption level $\epsilon$, and a data collecting distribution $\rho$, such that, with probability at least $1/4$, $L$ will find a no-better than $\Omega(Hd\epsilon)$-approximate NE policy pair $(\widetilde{\pi},\widetilde{\nu})$. That is, with a probability of at least $1/4$, we have, for every $s\in\mathcal{S}$:
    \begin{align*}
        \optgap(\widetilde{\pi},\widetilde{\nu},s) = \Omega(Hd\epsilon)~.
    \end{align*}
\end{statement}

\begin{proof}

We will construct an example to prove our statement. Consider the following Markov game $\mathcal{G}=(\mathcal{S},\mathcal{A},\mathcal{B},P,H,r,\gamma,s_0)$, with $\card{\mathcal{S}}=S$, $\card{\mathcal{A}}=A$, $\card{\mathcal{B}}=B$, deterministic transitions, episode length $H$ and initial state $s_0$, with $S\leq (AB)^{H/2}$. Note that, in this case, we have $d=SAB$. Here $r$ denotes the reward with respect to the $\max$ player. Assume that the transition dynamics follow a tree structure, that is, let $\mathcal{T}$ be a tree with nodes represented by states $s\in\mathcal{S}$ and edges represented by action tuples $(a,b)\in \mathcal{A}\times \mathcal{B}$, with root node $s_0$, such that, for every $(s,a,b)\in\mathcal{S}\times\mathcal{A}\times\mathcal{B}$, node $s$ is parent to node $\arg\max_{s'} P(s'\vert s,a,b)$. We denote by $p(s)$ the parent of node $s$. Moreover, assume that all states represented by the leaf nodes are self-absorbing states, i.e. the state does not change, no matter what action is taken. Let $q$ denote the depth of $\mathcal{T}$. Note that we have 
\begin{align*}
    q = O( \ceil{(\log_{AB}(S(AB-1)+1)-1})~.
\end{align*}

Now let us denote by $\mathcal{L}_i$ the subset of $\mathcal{S}$ containing the states represented by nodes in level $i$ of $\mathcal{T}$, for all $i \in \{ 0\} \cup [q]$, and let $s^j_i$ enumerate the states in level $\mathcal{L}_i$, for $j \in [(AB)^{i-1}]$. Let us define the reward function as follows. 
Fix a sequence of states $(s^*_0,s^*_1,\ldots,s^*_q) \in \mathcal{L}_0\times \ldots \times \mathcal{L}_q$, where $s^*_0=s_0$ and $s^*_i$ is a node in the $i$th level of $\mathcal{T}$ such that $P(s^*_i|\; s^*_{i-1},a,b_1)=1$, for all $i\in[q]$ and $a\in\mathcal{A}$. Furthermore, for all $s\neq s^*_{i-1}$ and $(a,b)\in \mathcal{A}\times \mathcal{B}$, we have $P(s^*_i|s,a,b)=0$, for all $i\in [q]$. Let $\alpha \in (0,1/3)$ and assume $(s^*_q,a_1,b_1)$ is the least represented state according to data collecting distribution $d^\rho$.

Then, for all $i \in \{ 0, 1, \ldots, q-1\}$, we define
\begin{equation*}
    r(s^*_i,a,b)=
    \begin{cases}
        \alpha & \text{if}\;\; a=a_1,b=b_1 \\
        2\alpha & \text{if}\;\; a=a_1,b\neq b_1 \\
        0 & \text{otherwise}
    \end{cases} \hspace{1cm}
    \text{and} \hspace{1cm}
     r(s^*_q,a,b)=
    \begin{cases}
        \alpha & \text{if}\;\; a=a_1,b=b_1 \\
        X & \text{if}\;\; a=a_1,b=b_2\\
        3\alpha & \text{if}\;\; a=a_1,b\in \mathcal{B}\setminus \{ b_1,b_2\} \\
        0 & \text{otherwise}
    \end{cases}
\end{equation*}
where $X$ is a Bernoulli random variable with parameter $2\alpha$.
On the other hand, for all $s\in \mathcal{S}\setminus \{ s^*_0,\ldots,s^*_q\}$, let
\begin{equation*}
    r(s,a,b)=
    \begin{cases}
        1 & \text{if}\;\; a=a_1\\
        0 & \text{otherwise}
    \end{cases}
\end{equation*}
Let us determine the value of the game in each state, using the method of backward induction. Set $V^*_{H+1}(s)=0$, for all $s\in\mathcal{S}$. Then, for all $s\in \mathcal{L}_q\setminus \{ s^*_q\}$, we have 
\begin{center}
\begin{tabular}{c|c c c c}
    $Q^*_H(s,\cdot,\cdot)$ & $b_1$ & $b_2$ & \ldots & $b_B$ \\
    \hline
     $a_1$ & $1$ & $1$ & \ldots & $1$ \\
     $a_2$ & $0$ & $0$ & \ldots & $0$ \\
     $\vdots$ & $\vdots$ & $\vdots$ & $\ddots$ & $\vdots$ \\
     $a_A$ & $0$ & $0$ & \ldots & $0$
\end{tabular} \hspace{1cm}
and \hspace{1cm}
\begin{tabular}{c|c c c c c}
    $Q^*_H(s^*_q,\cdot,\cdot)$ & $b_1$ & $b_2$ & $b_3$ & \ldots & $b_B$ \\
    \hline
     $a_1$ & $\alpha$ & $2\alpha$ & $3\alpha$ & \ldots & $3\alpha$ \\
     $a_2$ & $0$ & $0$ & $0$ & \ldots & $0$ \\
     $\vdots$ & $\vdots$ & $\vdots$ & $\vdots$ & $\ddots$ & $\vdots$ \\
     $a_A$ & $0$ & $0$ & $0$ &\ldots & $0$
\end{tabular}
\end{center}
Thus, we obtain $V^*_H(s)=1$, for all $s\in \mathcal{L}_d\setminus \{ s^*_q\}$ and $V^*_H(s^*_q)=\alpha$. Moreover, note that, for all $s\in \mathcal{L}_{d-1}\setminus \{ s^*_{q-1}\}$, we have
\begin{center}
\begin{tabular}{c|c c c c}
    $Q^*_{H-1}(s,\cdot,\cdot)$ & $b_1$ & $b_2$ & \ldots & $b_B$ \\
    \hline
     $a_1$ & $2$ & $2$ & \ldots & $2$ \\
     $a_2$ & $0$ & $0$ & \ldots & $0$ \\
     $\vdots$ & $\vdots$ & $\vdots$ & $\ddots$ & $\vdots$ \\
     $a_A$ & $0$ & $0$ & \ldots & $0$
\end{tabular} \hspace{1cm} 
and \hspace{1cm}
\begin{tabular}{c|c c c c c}
    $Q^*_{H-1}(s^*_q,\cdot,\cdot)$ & $b_1$ & $b_2$ & $b_3$ & \ldots & $b_B$ \\
    \hline
     $a_1$ & $2\alpha$ & $3\alpha$ & $4\alpha$ & \ldots & $4\alpha$ \\
     $a_2$ & $0$ & $0$ & $0$ & \ldots & $0$ \\
     $\vdots$ & $\vdots$ & $\vdots$ & $\vdots$ & $\ddots$ & $\vdots$ \\
     $a_A$ & $0$ & $0$ & $0$ &\ldots & $0$
\end{tabular}
\end{center}
and thus $V^*_{H-1}(s^*_q)=2\alpha$. Continuing in this fashion, we obtain 
\begin{align*}
    V^*_1(s_0) = H\alpha,
\end{align*}
where the first $q$ steps come from the trajectory $(s^*_0,\ldots,s^*_q)$, and the rest of the $H-q$ steps come from staying in state $s^*_q$.

Now let $\mathcal{G}'$ be a Markov game that is identical to $\mathcal{G}$, except for one difference. Let $r'$ denote the reward function of $\mathcal{G}'$. Then $r'(s^*_q,a_1,b_1) = r(s^*_q,a_1,b_1) + Ber(2\alpha)$, and $r'(s,a,b)=r(s,a,b)$, for all other state-action tuples. Then for $\mathcal{G}'$ we have
\begin{center}
\begin{tabular}{c|c c c}
    $\widetilde{Q}^*_H(s,\cdot,\cdot)$ & $b_1$ & \ldots & $b_B$ \\
    \hline
     $a_1$ & $1$ & \ldots & $1$ \\
     $a_2$ & $0$ & \ldots & $0$ \\
     $\vdots$ & $\vdots$ & $\ddots$ & $\vdots$ \\
     $a_A$ & $0$ & \ldots & $0$
\end{tabular} \hspace{1cm}
and \hspace{1cm}
\begin{tabular}{c|c c c c c}
    $\widetilde{Q}^*_H(s^*_q,\cdot,\cdot)$ & $b_1$ & $b_2$ & $b_3$ & \ldots & $b_B$ \\
    \hline
     $a_1$ & $3\alpha$ & $2\alpha$ & $2\alpha$ & \ldots & $2\alpha$ \\
     $a_2$ & $0$ & $0$ & $0$ & \ldots & $0$ \\
     $\vdots$ & $\vdots$ & $\vdots$ & $\vdots$ & $\ddots$ & $\vdots$ \\
     $a_A$ & $0$ & $0$ & $0$ &\ldots & $0$
\end{tabular}
\end{center}
where $\widetilde{Q}(\cdot,\cdot,\cdot)$ denotes the matrices of $Q$-values of NE policies for $\mathcal{G}'$. Note that the trajectory traversed by the NE policy pair in $\mathcal{G}'$ is still $(s^*_0,\ldots,s^*_{q-1},s^*_q)$. However, when at state $s^*_q$, the NE policy is $(a_1,b_2)$, instead of $(a_1,b_1)$. Thus, we have
\begin{align*}
    \widetilde{V}^*_1(s_0) = (2H-d)\alpha,
\end{align*}
since the system will stay in state $s^*_q$ for $H-q$ steps, until the episode ends. Note that no policy pair can be simultaneously optimal in both games. In the worst case, a policy pair which is a NE in $\mathcal{G}'$ will incur a suboptimality gap of
\begin{align*}
    (H-q)\alpha = (H - \ceil{(\log_{AB}(S(AB-1)+1)-1})\alpha \geq \Omega(H\alpha),
\end{align*}
in $\mathcal{G}$, where the second inequality follows from the fact that $S\leq (AB)^{H/2}$.

Now, let $\alpha = SAB\epsilon /2$. Since $(s^*_q,a_1,b_1)$ is the least represented state with respect to $d^\rho$, by pigeon-hole principle, we must have $d^\rho(s^*_d,a_1,b_1)\leq 1/SAB$. Assume the adversary uses all its budget only to perturb the reward of this state-action tuple. Concretely, if the game is $\mathcal{G}$, then the adversary perturbs it into $\mathcal{G}'$ by adding $Ber(SAB\epsilon)$ to $r(s^*_q,a_1,b_1)$. 

With probability at least $1/2$, the number of times $(s^*_q,a_1,b_1)$ is counted in a dataset with $KH$ tuples is no more than $KH/SAB$, since $(s^*_q,a_1,b_1)$ is the least represented tuple. Conditioned on this, with probability at least $1/2$, the reward seen from $(s^*_q,a_1,b_1)$ is $2SAB\epsilon$. Thus, perturbing the reward of $(s^*_q,a_1,b_1)$ at least $KH\epsilon$ times is enough to make one of the games indistinguishable from the other one. Thus, the agent will inevitably incur 
\begin{align*}
    \optgap (\pi,\nu,s) \geq \Omega (HSAB\epsilon).
\end{align*}
\end{proof}

\subsection{Proof of Theorem \ref{thm:linear_uniform_coverage}}

The main result of this section relies on the RLS oracle guarantee stated below.

\begin{thm}\label{thm:rls_guarantee} \citep{zhang2022corruption}
Given an $\epsilon$-corrupted dataset $D=\{ x_i,y_i\}_{i \in [n]}$, where the clean data is generated as $\widetilde{x}_i \sim \beta$, $\mathbb{P}(\norm{\widetilde{x}_i}\leq 1) = 1$, $\widetilde{y}_i = \widetilde{x}_i^\top \omega^*+\xi_i$, where $\xi_i$ is zero-mean $\sigma^2$-variance sub-Gaussian random noise, then a robust least square estimator returns an estimator $\omega$ such that, if $\mathbb{E}_\beta[xx^\top]\succeq \kappa I$, for some strictly positive constant $\kappa$, then with probability at least $1-\delta$, we have 
\begin{itemize}
    \item If $\mathbb{E}_\beta[xx^\top]\succeq \kappa I$, then with probability at least $1-\delta$, we have $$\norm{\omega^* - \omega}_2 \leq c_1(\delta) \cdot \left( \sqrt{\frac{\sigma^2poly(d)}{\kappa^2 n}} + \frac{\sigma}{\kappa}\epsilon \right);$$
    \item With probability at least $1-\delta$, we have $$\mathbb{E}\left[ \norm{\widetilde{x}^\top(\omega^* - \omega)}^2_2\right] \leq c_2(\delta) \cdot \left( \frac{\sigma^2poly(d)}{n} + \sigma^2\epsilon \right),$$
\end{itemize}
where $c_1$ and $c_2$ hide constants and $polylog (1/\delta)$ terms.
\end{thm}

Using the above guarantee, we are now ready to prove Theorem \ref{thm:linear_uniform_coverage}. First, we recall its statement.

\begin{statement}
        Suppose that Assumption \ref{asn:uniform-phi-coverage} holds on an $\epsilon$-corrupted dataset $D$ corresponding to a linear Markov game. Then, given $\delta > 0$, with probability at least $1-\delta$, RLS-PMVI with bonus term $\Gamma_h(s,a,b)=0$ achieves suboptimality gap upper bounded by
    \begin{align*}
         \widetilde{O} \left( \sqrt{\frac{H(H+\gamma)^2poly(d)}{\kappa^2 K}} + \frac{H(H+\gamma)}{\kappa}\epsilon \right)~.
    \end{align*}
\end{statement}

\begin{proof}
Let us define the error in Theorem \ref{thm:rls_guarantee} as
\begin{equation}
    \mathcal{E}_1(\epsilon,K,D,H,\sigma) = c_1(\delta) \cdot \left( \sqrt{\frac{\sigma^2poly(d)}{\kappa^2 K}} + \frac{\sigma}{\kappa}\epsilon \right)~.
\end{equation}
First, we provide upper and lower bounds on the Bellman error. 
\begin{lem}\label{lem:uniform_coverage_bellman_error}
    Given the tuple $(s,a,b)$, let us define the Bellman error as
    \begin{align*}
        \underline{\iota}_h (s,a,b) & = (\mathbb{B}_h\underline{V}_{h+1})(s,a,b) - \underline{Q}_h(s,a,b)~, \\
        \overline{\iota}_h (s,a,b) & = (\mathbb{B}_h\overline{V}_{h+1})(s,a,b) - \overline{Q}_h(s,a,b)~,
    \end{align*}
    for all $h\in[H]$. Then, with probability at least $1-\delta$, we have
\begin{align*}
    -\mathcal{E}_1(\epsilon,K,D,H,\sigma) & \leq \underline{\iota}_h(s,a,b) \leq \mathcal{E}_1(\epsilon,K,D,H,\sigma)~, \\
    -\mathcal{E}_1(\epsilon,K,D,H,\sigma) & \leq -\overline{\iota}_h(s,a,b) \leq \mathcal{E}_1(\epsilon,K,D,H,\sigma)~,
\end{align*}
for all $(s,a,b)\in\mathcal{S}\times\mathcal{A}\times\mathcal{B}$ and $h\in[H]$.
\end{lem}
\begin{proof}
We start by deriving upper bounds on the Bellman error. Note that, with probability at least $1-\delta$, Theorem \ref{thm:rls_guarantee} and Assumption \ref{asmp_linearity} imply
\begin{align}\label{eq:uniform_coverage_01}
    \vert \phi^\top \underline{\omega}_h - (\mathbb{B}_h\underline{V}_{h+1})(s,a,b)\vert \leq \norm{\phi(s,a,b)}_2 \norm{\underline{\omega}_h - \underline{\omega}^*_h}_2 \leq \mathcal{E}_1(\epsilon,K,D,H,\sigma)
\end{align}
 We show the first case. 
Note that we have
\begin{equation*}
    \underline{Q}_h(s,a,b) = \max\{ 0, \phi^\top \underline{\omega}_h \} \geq \phi^\top \underline{\omega}_h , 
\end{equation*}
which, together with Equation \eqref{eq:uniform_coverage_01}, imply 
\begin{align*}
    \underline{\iota}_h(s,a,b) & = (\mathbb{B}_h\underline{V}_{h+1})(s,a,b) -  \underline{Q}_h(s,a,b) \\
            & \leq (\mathbb{B}_h\underline{V}_{h+1})(s,a,b) - \phi^\top \underline{\omega}_h \\
            & \leq \mathcal{E}_1(\epsilon,K,D,H,\sigma)~.
\end{align*}
For the lower bound, consider two cases. First, if $\phi^\top \underline{\omega}_h  \leq 0$, then we have
\begin{align*}
    \underline{\iota}_h(s,a,b) & = (\mathbb{B}_h\underline{V}_{h+1})(s,a,b) -  \underline{Q}_h(s,a,b) \\
            & = (\mathbb{B}_h\underline{V}_{h+1})(s,a,b) - 0 \\
            & \geq 0.
\end{align*}
On the other hand, if $\phi^\top \underline{\omega}_h  \geq 0$, then we have
\begin{align*}
    \underline{\iota}_h(s,a,b) & = (\mathbb{B}_h\underline{V}_{h+1})(s,a,b) -  \underline{Q}_h(s,a,b) \\
            & = (\mathbb{B}_h\underline{V}_{h+1})(s,a,b) - \phi^\top \underline{\omega}_h \\
            & \geq -\mathcal{E}_1(\epsilon,K,D,H,\sigma) ~.
\end{align*}
\end{proof}

Next, we consider the relationship between the estimated value function and the true ones based on the best responses.
\begin{lem}\label{lem:uniform_coverage_value_function}
If the bounds on the Bellman error given above hold, then, for any $s\in\mathcal{S}$, we have
\begin{align*}
    \underline{V}_h(s) \leq V^{\widehat{\pi},*}_h(s) + \mathcal{E}_1 \;\; \text{and} \;\; V^{*,\widehat{\nu}}_h(s) - \mathcal{E}_1 \leq \overline{V}_h(s)~,
\end{align*}
    where we omit the dependence of $\mathcal{E}_1$ on the relevant variables for brevity.
\end{lem}
\begin{proof}  
We prove the left inequality. The right follows similar arguments. We use backward induction. For $h=H+1$, we have $\underline{V}_h(s)=V^{\hat{\pi},*}_h(s)=0$. We assume the inequality hold for $h+1$ and prove it for $h$. We have
\begin{align*}
    V^{\hat{\pi},*}_h(s) - \underline{V}_h(s) & = \mathbb{E}_{a\sim \hat{\pi}, b\sim *}\left[ Q^{\hat{\pi},*}_h(s,a,b)\right] - \mathbb{E}_{a\sim\hat{\pi}, b\sim\hat{\nu}} \left[ \underline{Q}_h(s,a,b)\right] \\
        & = \mathbb{E}_{a\sim \hat{\pi}, b\sim *}\left[ Q^{\hat{\pi},*}_h(s,a,b) - \underline{Q}_h(s,a,b)\right] + \left( \mathbb{E}_{a\sim \hat{\pi}, b\sim *}\left[ \underline{Q}_h(s,a,b) \right] - \mathbb{E}_{a\sim\hat{\pi}, b\sim\hat{\nu}} \left[ \underline{Q}_h(s,a,b)\right] \right)~.
\end{align*}
First, note that
\begin{align*}
    Q^{\hat{\pi},*}_h(s,a,b) - \underline{Q}_h(s,a,b) = \mathbb{B}_h\left( V^{\hat{\pi},*}_{h+1}(s,a,b) - \underline{V}_{h+1}(s,a,b)\right) + \underline{\iota}_h(s,a,b) \geq - \mathcal{E}_1~,
\end{align*}
where the second inequality follows from Lemma \ref{lem:uniform_coverage_bellman_error} and the induction assumption. Furthermore, by the NE value property, we have
\begin{align*}
    \mathbb{E}_{a\sim \hat{\pi}, b\sim *}\left[ \underline{Q}_h(s,a,b) \right] - \mathbb{E}_{a\sim\hat{\pi}, b\sim\hat{\nu}} \left[ \underline{Q}_h(s,a,b)\right] \geq 0~.
\end{align*}
Thus, we obtain $\underline{V}_h(s)  \leq V^{\hat{\pi},*}_h(s) + \mathcal{E}_1$.
\end{proof}
Next, we prove a result that gives us an upper bound on $\sigma^2$ in terms of the reward variance and $H$. 
\begin{lem}\label{lem:variance_bound}
    We have $Var (\xi_h^\tau | s^\tau_h,a^\tau_h,b^\tau_h) = \sigma^2  \leq (H+\gamma)^2$. 
\end{lem}

\begin{proof}
    We consider $\underline{V}_{h+1}$. The case for $\overline{V}_{h+1}$ is similar. We have
    \begin{align*}
        Var (\xi_h^\tau | s^\tau_h,a^\tau_h,b^\tau_h) & = \mathbb{E} \left[ (r^\tau_h + \underline{V}_{h+1}(s^\tau_{h+1}) - (\mathbb{B}_h\underline{V}_{h+1})(s^\tau_h,a^\tau_h,b^\tau_h) )^2\vert s^\tau_h,a^\tau_h,b^\tau_h \right] \\
            & = \mathbb{E}\left[ (r^\tau_h + \underline{V}_{h+1}(s^\tau_{h+1}) - \mathbb{E}[r^\tau_h + \underline{V}_{h+1}(s^\tau_{h+1}) \vert s^\tau_h,a^\tau_h,b^\tau_h ] )^2 \vert s^\tau_h,a^\tau_h,b^\tau_h  \right] \\
            & = \mathbb{E}\left[ (r^\tau_h - \mathbb{E}[r^\tau_h\vert s^\tau_h,a^\tau_h,b^\tau_h])^2 \vert s^\tau_h,a^\tau_h,b^\tau_h\right]
             + \mathbb{E}\left[ (\underline{V}_{h+1}(s^\tau_{h+1}) - \mathbb{E}[\underline{V}_{h+1}(s^\tau_{h+1})\vert s^\tau_h,a^\tau_h,b^\tau_h])^2 \vert s^\tau_h,a^\tau_h,b^\tau_h \right] \\
             & \quad + 2\mathbb{E}\left[ (r^\tau_h - \mathbb{E}[r^\tau_h\vert s^\tau_h,a^\tau_h,b^\tau_h])(\underline{V}_{h+1}(s^\tau_{h+1}) - \mathbb{E}[\underline{V}_{h+1}(s^\tau_{h+1})\vert s^\tau_h,a^\tau_h,b^\tau_h]) \right] \\
             & = Var(r^\tau_h\vert s^\tau_h,a^\tau_h,b^\tau_h) + Var(\underline{V}_{h+1}(s^\tau_{h+1})\vert s^\tau_h,a^\tau_h,b^\tau_h) \\
             & \quad + 2\sqrt{\mathbb{E}\left[ (r^\tau_h - \mathbb{E}[r^\tau_h\vert s^\tau_h,a^\tau_h,b^\tau_h])^2\right] \mathbb{E}\left[(\underline{V}_{h+1}(s^\tau_{h+1}) - \mathbb{E}[\underline{V}_{h+1}(s^\tau_{h+1})\vert s^\tau_h,a^\tau_h,b^\tau_h])^2 \right] } \\
             & \leq Var(r^\tau_h\vert s^\tau_h,a^\tau_h,b^\tau_h) + Var(\underline{V}_{h+1}(s^\tau_{h+1})\vert s^\tau_h,a^\tau_h,b^\tau_h) \\ & \quad + 2\sqrt{Var(r^\tau_h\vert s^\tau_h,a^\tau_h,b^\tau_h) Var(\underline{V}_{h+1}(s^\tau_{h+1})\vert s^\tau_h,a^\tau_h,b^\tau_h) }\\
             & = \left( \sqrt{Var(r^\tau_h\vert s^\tau_h,a^\tau_h,b^\tau_h)} + \sqrt{Var(\underline{V}_{h+1}(s^\tau_{h+1})\vert s^\tau_h,a^\tau_h,b^\tau_h)} \right)^2 \leq (\gamma + H)^2~,
    \end{align*}
    where the fourth equality uses Cauchy-Schwarz and the last one uses the fact that $0\leq \underline{V}_{h+1}(s)\leq H$, for all $h\in[H]$, $s\in\mathcal{S}$.
\end{proof}

Next, we state the following well-known result which will help us express the suboptimality gap in terms of quantities that we can control. For a proof, see \citep{cai2020provably}.

\begin{lem}\label{lem:value_difference}
    (Value difference lemma) Given an MG $(\mathcal{S},\mathcal{A},\mathcal{B},r,H)$, let $\hat{\pi}\otimes\hat{\nu} = \{\hat{\pi}_h\otimes\hat{\nu}_h: \mathcal{S} \rightarrow \Delta(\mathcal{A})\times\Delta(\mathcal{B})\}_{h\in[H]}$ be a product policy, $(\pi,\nu)$ be a policy pair, and $(\widehat{Q}_h)_{h\in[H]}$ be any estimated $Q$-functions. For any $h\in[H]$ we define the estimated value function $\widehat{V}_h:\mathcal{S}\rightarrow\mathbb{R}$ by setting $\widehat{V}_h(s) = \langle\widehat{Q}_h(s,\cdot,\cdot),\hat{\pi}_h(\cdot|s)\otimes\hat{\nu}(\cdot|s)\rangle$, for all $s\in\mathcal{S}$. Then, for all $s\in\mathcal{S}$, we have
    \begin{align*}
        \widehat{V}_1(s)-V^{\pi,\nu}_1(s) & = \sum^H_{h=1} \mathbb{E}_{\pi,\nu}\left[ \langle\widehat{Q}_h(s,\cdot,\cdot),\hat{\pi}_h(\cdot|s)\otimes\hat{\nu}(\cdot|s) - \pi(\cdot|s)\otimes\nu(\cdot|s)\rangle \vert s_1=s\right] \\
            & \quad + \sum^H_{h=1}\mathbb{E}_{\pi,\nu}\left[ \widehat{Q}_h(s_h,a_h,b_h) - (\mathbb{B}_h V_{h+1})(s_h,a_h,b_h)\vert s_1=s\right]~.
    \end{align*} 
\end{lem}

Now, the first term in the lemma above can be bounded by $0$ (see Lemma A.3 of \citep{zhong2022pessimistic}). Thus, we obtain
\begin{align}\label{eq:sub_opt_gap_alternative}
    \widehat{V}_1(s)-V^{\pi,\nu}_1(s)  \leq \sum^H_{h=1}\mathbb{E}_{\pi,\nu}\left[ \widehat{Q}_h(s_h,a_h,b_h) - (\mathbb{B}_h V_{h+1})(s_h,a_h,b_h)\vert s_1=s\right]~.
\end{align}

Now we write the suboptimality gap as:
\begin{align}\label{eq:sub_opt_gap_decomposition}
    \optgap (\widehat{\pi},\widehat{\nu},s) = V^{*,\widehat{\nu}}_1(s) - V^{\widehat{\pi},*}_1(s) = \left( V^{*,\widehat{\nu}}_1(s) - V^*_1(s)\right) + \left( V^*_1(s) - V^{\widehat{\pi},*}_1(s)\right) ~.
\end{align}
For the first term, we have
\begin{equation*}
    V^{*,\widehat{\nu}}_1(s) - V^*_1(s) \leq \overline{V}_1(s) - V^*_1(s) + \mathcal{E}_1 \leq \overline{V}_1(s) - V^{\pi',\nu^*}_1(s) + \mathcal{E}_1~,
\end{equation*}
for some $\pi'$, where the first inequality follows from Lemma \ref{lem:uniform_coverage_value_function} and the second inequality follows from the fact that $(\pi^*,\nu^*)$ is an NE policy. Similarly, we have that 
$V^*_1(s) - V^{\widehat{\pi},*}_1(s) \leq V^{\pi^*,\nu'}_1(s) - \underline{V}_1(s) + \mathcal{E}_1$, for some $\nu'$. 
Then, Equation \eqref{eq:sub_opt_gap_alternative}, Lemma \ref{lem:uniform_coverage_bellman_error} and Lemma \ref{lem:value_difference} imply
\begin{align*}
    V^{*,\widehat{\nu}}_1(s) - V^*_1(s) & \leq \overline{V}_1(s) - V^{\pi',\widehat{\nu}}_1(s) + \mathcal{E}_1(\epsilon,K,H,d,\sigma) \\
    & \leq \sum^H_{h=1} \mathbb{E}_{\pi',\nu^*}\left[ -\overline{\iota}_h(s_h,a_h,b_h) \right] + H \mathcal{E}_1(\epsilon,K,H,d,\sigma)  \\
        & \leq 2H \mathbb{E}_{\pi',\nu^*}\left[ \mathcal{E}_1(\epsilon,K,H,d,\sigma)  \vert s_1=s \right] \\
        & \leq 2H c_1(\delta) \cdot \left( \sqrt{\frac{\sigma^2poly(d)}{\kappa^2 K}} + \frac{\sigma}{\kappa}\epsilon \right) \\
        & \leq 2H c_1(\delta) \cdot \left( \sqrt{\frac{(H+\gamma)poly(d)}{\kappa^2 K}} + \frac{H+\gamma}{\kappa}\epsilon \right)~,
\end{align*}
where the last inequality follows from Lemma \ref{lem:variance_bound}. Similarly, we have
\begin{align*}
    V^*_1(s) - V^{\widehat{\pi},*}_1(s) & \leq 2H c_1(\delta) \cdot \left( \sqrt{\frac{(H+\gamma)poly(d)}{\kappa^2 K}} + \frac{H+\gamma}{\kappa}\epsilon \right)~.
\end{align*}
Finally, we obtain
\begin{align*}
    \optgap (\widehat{\pi},\widehat{\nu},s) & \leq O \left( \sqrt{\frac{H(H+\gamma)^2poly(d)}{\kappa^2 K}} + \frac{H(H+\gamma)}{\kappa}\epsilon \right)~.
\end{align*}
\end{proof}

\subsection{Proof of Theorem \ref{thm:linear_case}}

Given clean dataset $\widetilde{D} = \{ (\tilde{x}_1,\tilde{y}_1),\ldots,(\tilde{x}_K,\tilde{y}_K)\}$, consider the observation model 
\begin{align*}
    \tilde{y}_i = \langle \omega^*, \tilde{x}_i \rangle + \xi_i~,
\end{align*}
where $\xi_i$ are $\gamma^2$-sub-Gaussian independent noise variables and $\omega^*$ is the true regressor with $\norm{\omega^*}\leq R$, for some $R < \infty$, and $\norm{\tilde{x}_i}_2\leq 1$, for $i\in[K]$. Furthermore, given $\epsilon < 1/2$, assume that $\epsilon$-fraction of the outcomes in $\widetilde{D}$ are corrupted and let $D = \{ (x_1,y_1),\ldots,(x_K,y_K)\}$ denote the corrupted dataset. Formally, we assume that the covariates remain clean, that is, $x_i=\tilde{x}_i$, for $i\in [K]$, and that, for any $i\in [K]$, a coin is flipped with success rate $\epsilon$ to determine whether $\tilde{y}_i$ is corrupted into $y_i$ or not. 

Furthermore, let 
\begin{align*}
    \Sigma = \frac{1}{K}\sum^K_{k=1}x_k x_k^T
\end{align*}
denote the covariance matrix of $D$. The following result \citep{chen2022online} provides error norm bounds of the regressor estimated using the SCRAM method on the corrupted dataset $D$.


\begin{thm}\label{thm:scram_guarantee}
Let $0<\epsilon<0.5$ be an upper bound on the contamination level, and suppose $K$ satisfies $K\geq  O(\log (\min \{ K,d\} )/\epsilon)$. Then, given $\delta \in (0,1)$, there exists a $poly(K,d)$ algorithm which takes as input the dataset $D$ and, with probability at least $1-\delta$, outputs a vector $\omega$ that satisfies

\begin{align*}
    \norm{\omega^* - \omega}_{\Sigma} \leq O\left( \epsilon \gamma \log(1/\epsilon) + \min \left\{ \gamma \sqrt{\frac{d+\log(1/\delta)}{K}}, (R\gamma)^{1/2}\sqrt[4]{\frac{\log(1/\delta)}{K}} \right\} \right)~.
\end{align*}

\end{thm}



For every $h\in[H]$, we define the datasets 
\begin{align*}
    \widetilde{D}_{min}(h) = \{ \underbrace{\phi (\tilde{s}^\tau_h,\tilde{a}^\tau_h,\tilde{b}^\tau_h)}_\text{covariates}, \underbrace{\tilde{r}^\tau_h + \underline{V}_{h+1}(\tilde{s}^\tau_{h+1})}_\text{clean obs.} \}^K_{\tau=1}
\end{align*}
 and 
\begin{align*}
    \widetilde{D}_{max}(h) = \{ \phi (\tilde{s}^\tau_h,\tilde{a}^\tau_h,\tilde{b}^\tau_h), \tilde{r}^\tau_h + \overline{V}_{h+1}(\tilde{s}^\tau_{h+1}) )\}^K_{\tau=1}~.
\end{align*}
Similarly, the partitions of the corrupted data are defined as
\begin{align*}
     D_{min}(h) = \{ \underbrace{\phi (s^\tau_h,a^\tau_h,b^\tau_h)}_\text{covariates}, \underbrace{r^\tau_h + \underline{V}_{h+1}(s^\tau_{h+1})}_\text{$\epsilon$-corrupted obs.} )\}^K_{\tau=1}
\end{align*}
and
\begin{align*}
     D_{max}(h) = \{ \phi (s^\tau_h,a^\tau_h,b^\tau_h), r^\tau_h + \overline{V}_{h+1}(s^\tau_{h+1}) )\}^K_{\tau=1}~.
\end{align*}
Note that we assume $\phi (\tilde{s}^\tau_h,\tilde{a}^\tau_h,\tilde{b}^\tau_h) = \phi (s^\tau_h,a^\tau_h,b^\tau_h)$, for all $\tau \in [K]$. Assumption \ref{asmp_linearity} implies that there exist weight vectors $\underline{\omega}^*_h, \overline{\omega}^*_h \in \mathbb{R}^d$ such that we have
\begin{align}
    \phi (s^\tau_h,a^\tau_h,b^\tau_h)^\top \underline{\omega}^*_h + \xi^\tau_h  = (\mathbb{B}_h\underline{V}_{h+1})(s^\tau_h,a^\tau_h,b^\tau_h) + \xi^\tau_h  = \left( \tilde{r}^\tau_h + \underline{V}_{h+1}(\tilde{s}^\tau_{h+1})\right) \label{eq:weights_of_bellman_1}
\end{align}
and 
\begin{align}
    \phi (s^\tau_h,a^\tau_h,b^\tau_h)^\top \overline{\omega}^*_h + \xi^\tau_h = (\mathbb{B}_h\overline{V}_{h+1})(s^\tau_h,a^\tau_h,b^\tau_h) + \xi^\tau_h = \left( \tilde{r}^\tau_h + \overline{V}_{h+1}(\tilde{s}^\tau_{h+1})\right)  ~, \label{eq:weights_of_bellman_2}
\end{align}
where $\xi^\tau_h$ are zero-mean $\gamma^2$-subGaussian random variables. 

Now, let us define the covariance matrices, for all $h\in[H]$, as:
\begin{align*}
    \Sigma_h = \frac{1}{K}\sum^K_{\tau=1}\phi(s^\tau_h,a^\tau_h,b^\tau_h)\phi (s^\tau_h,a^\tau_h,b^\tau_h)^\top ~.
\end{align*}
Note that $\Sigma_h$ depend on the triples $(s^\tau_h,a^\tau_h,b^\tau_h)$, $\tau \in [K]$, which are not changed under corruption. Thus, the covariance matrices are clean. However, the observations on both $\widetilde{D}_{min}(h)$ and $\widetilde{D}_{max}(h)$, for all $h\in[H]$, are $\epsilon$-corrupted. Then, Theorem \ref{thm:scram_guarantee} implies the following result.
\begin{cor}\label{cor:scram_guarantee}
    Under the conditions of Theorem \ref{thm:scram_guarantee}, there exist a $poly(K,d)$ algorithm that returns a sequence of estimators $(\underline{\omega}_h,\overline{\omega}_h)^H_{h=1}$ such that, given $\delta >0$, the following inequalities are satisfied, for all $h\in[H]$, with probability at least $1-\delta/2$:
    \begin{align}
        \norm{\underline{\omega}^*_h - \underline{\omega}_h}_{\Sigma_h} & \leq O\left( \gamma \epsilon \log(1/\epsilon) + \min \left\{ \gamma \sqrt{\frac{d+\log (8H/\delta)}{K}}, \sqrt{H\gamma}\sqrt[4]{\frac{d\log(8H/\delta)}{K}} \right\} \right)~, \label{eq:corruption_error_bound_1}\\
        \norm{\overline{\omega}^*_h - \overline{\omega}_h}_{\Sigma_h} & \leq O\left( \gamma \epsilon \log(1/\epsilon) + \min \left\{ \gamma \sqrt{\frac{d+\log (8H/\delta)}{K}}, \sqrt{H\gamma}\sqrt[4]{\frac{d\log(8H/\delta)}{K}} \right\} \right)~. \label{eq:corruption_error_bound_2}
    \end{align}
    We will denote the right-hand side bound on the errors of norms by $\mathcal{E}(\epsilon,K,H,d)$, as shorthand notation. 
\end{cor}

\begin{proof}
The result follows immediately from Theorem \ref{thm:scram_guarantee} by applying the union bound over $2H$ events, and also by noting that $\norm{\underline{\omega}^*_h}, \norm{\overline{\omega}^*_h} \leq H\sqrt{d}$, for all $h\in[H]$, by Lemma E.1 of \citep{zhong2022pessimistic}.
\end{proof}

Next, we prove the upper bounds on the Bellman error in terms of corruption level and bonus term. Given $(s,a,b)\in \mathcal{S}\times \mathcal{A}\times \mathcal{B}$ and $h\in[H]$, recall that the bonus term is defined as 
\begin{align*}
    \Gamma_h(s,a,b) =\left(\sqrt{K} \mathcal{E}(\epsilon,K,H,d) + 2H\sqrt{d} \right) \norm{\phi(s,a,b)}_{\Lambda_h^{-1}}~,
\end{align*}
where 
\begin{align*}
    \Lambda_h = \sum^K_{\tau=1}\phi(s^\tau_h,a^\tau_h,b^\tau_h)\phi (s^\tau_h,a^\tau_h,b^\tau_h)^\top + I~.
\end{align*}
Note that $\Lambda_h$ is positive definite, and hence, invertible, since $\Sigma_h$ is positive semi-definite.

From here on, let us use the following notation for ease of presentation. For a given $(s,a,b)$ and $h\in[H]$, let
\begin{align*}
    \phi := \phi(s,a,b), \;\; \text{and} \;\; \phi_h := \phi(s_h,a_h,b_h)~.
\end{align*}

\begin{lem}\label{lem:bellman_error}
    Given tuple $(s,a,b)$, let $\underline{\iota}_h(s,a,b)$ and $\overline{\iota}_h(s,a,b)$ be defined as in Lemma \ref{lem:uniform_coverage_bellman_error}. Then, with probability at least $1-\delta$, we have 
    \begin{align}
        0 & \leq \underline{\iota}_h(s,a,b) \leq 2\Gamma_h(s,a,b) ~,\label{eq:bellman_error_1}\\
        0 & \leq -\overline{\iota}_h(s,a,b) \leq 2\Gamma_h (s,a,b) \label{eq:bellman_error_2}~,
    \end{align}
    for all $(s,a,b) \in \mathcal{S}\times \mathcal{A}\times \mathcal{B}$.
\end{lem}

\begin{proof}
We will prove the first inequality, and the second will follow by symmetry of argument.
Let $\underline{\omega}^*_h$ be defined as in Equation \ref{eq:weights_of_bellman_1}.
First, note that Corollary \ref{cor:scram_guarantee} implies
\begin{align*}
    \norm{\underline{\omega}_h - \underline{\omega}^*_h}^2_{\Lambda_h} & = (\underline{\omega}_h - \underline{\omega}^*_h)^\top \Lambda_h (\underline{\omega}_h - \underline{\omega}^*_h) \\
            & = (\underline{\omega}_h - \underline{\omega}^*_h)^\top (K\Sigma_h + I)(\underline{\omega}_h - \underline{\omega}^*_h) \\
            & = K(\underline{\omega}_h - \underline{\omega}^*_h)^\top \Sigma_h (\underline{\omega}_h - \underline{\omega}^*_h) + (\underline{\omega}_h - \underline{\omega}^*_h)^\top I (\underline{\omega}_h - \underline{\omega}^*_h) \\
            & \leq K\norm{\underline{\omega}_h - \underline{\omega}^*_h}^2_{\Sigma_h} + 4H^2d \\
            & \leq K\mathcal{E}(\epsilon,K,H,d)^2 + 4H^2d
\end{align*}
where the first inequality comes from the fact that $\norm{\underline{\omega}^*_h}_2 \leq H\sqrt{d}$, from Lemma E.1 of \citep{zhong2022pessimistic} and also $\norm{\underline{\omega}_h}_2 \leq H\sqrt{d}$, by design of SCRAM \citep{chen2022online}. Thus, taking the square roots of both sides, we obtain
\begin{align}
    \norm{\underline{\omega}_h - \underline{\omega}^*_h}_{\Lambda_h} & \leq \sqrt{K\mathcal{E}(\epsilon,K,H,d)^2 + 4H^2d} 
     \leq \sqrt{K}\mathcal{E}(\epsilon,K,H,d) + 2H\sqrt{d}~. \label{eq:bellman_error_001}
\end{align}

Then, with probability at least $1-\delta$, we have
\begin{align}
   \vert \phi^\top\underline{\omega}_h - (\mathbb{B}_h\underline{V}_{h+1})(s,a,b)\vert & = \vert \phi^\top (\underline{\omega}_h - \underline{\omega}^*_h)\vert \nonumber 
         \leq \norm{\underline{\omega}_h - \underline{\omega}^*_h}_{\Lambda_h} \norm{\phi}_{\Lambda^{-1}_h}  
         \leq \left(\sqrt{K} \mathcal{E}(\epsilon,K,H,d) + 2H\sqrt{d} \right) \norm{\phi}_{\Lambda^{-1}_h}~,
\end{align}
where the first inequality follows from extended Cauchy-Schwarz.

Note that each of the bounds on the right-hand side of the last inequality holds with probability at least $1-\delta/2$, thus, by union bound, the inequality above holds with probability at least $1-\delta$. Therefore, we obtain
\begin{align*}
    \phi^\top \underline{\omega}_h - \Gamma_h(s,a,b) \leq \mathbb{B}_h\underline{V}_{h+1} \leq H-h+1~,
\end{align*}
where in the last inequality we have used the fact that $|r_h|\leq 1$ and $|\underline{V}_{h+1}(s)|\leq H-h$. Furthermore,
\begin{align*}
    \underline{Q}_h(s,a,b) = \max \{ 0, \phi^\top\underline{\omega}_h - \Gamma_h(s,a,b)\} \geq \phi^\top \underline{\omega}_h - \Gamma_h(s,a,b) ~.
\end{align*}
Thus, we have
\begin{align*}
    \underline{\iota}_h(s,a,b) & = (\mathbb{B}_h\underline{V}_{h+1})(s,a,b) - \underline{Q}_h(s,a,b) \\
                    & \leq (\mathbb{B}_h\underline{V}_{h+1})(s,a,b) - \phi^\top \underline{\omega}_h + \Gamma_h(s,a,b)  \\
                    & \leq 2\Gamma_h(s,a,b)~.
\end{align*}
To prove the non-negativity of $\underline{\iota}_h(s,a,b)$, we consider two cases. First, if $\phi^\top\underline{\omega}_h - \Gamma_h(s,a,b) \geq 0$, then we have 
\begin{align*}
    \underline{\iota}_h(s,a,b) & = (\mathbb{B}_h\underline{V}_{h+1})(s,a,b) - \underline{Q}_h(s,a,b) \\
            & = (\mathbb{B}_h\underline{V}_{h+1})(s,a,b) -0  \geq 0~,
\end{align*}
since $\underline{V}_{h+1}(s) \in [0,H]$ and $r(s,a,b)\in[0,1]$. On the other hand, if $\phi^\top\underline{\omega}_h - \Gamma_h(s,a,b) \leq 0$, we have
\begin{align*}
     \underline{\iota}_h(s,a,b) & = (\mathbb{B}_h\underline{V}_{h+1})(s,a,b) - \underline{Q}_h(s,a,b) \\
            & = (\mathbb{B}_h\underline{V}_{h+1})(s,a,b) - \phi^\top\underline{\omega}_h + \Gamma_h(s,a,b)  \geq  0~.
\end{align*}
\end{proof}

Next, we give bounds on the best response values in terms of the estimated value functions.

\begin{lem}\label{lem:value_function_bounds}
    If Equations \eqref{eq:bellman_error_1} and \eqref{eq:bellman_error_2} in Lemma \ref{lem:bellman_error} hold, then, for any $s\in\mathcal{S}$, we have
    \begin{align*}
        \underline{V}_h(s) \leq V^{\hat{\pi},*}_h(s), \;\; \text{and} \;\; V^{*,\hat{\nu}}_h(s) \leq \overline{V}_h(s) ~.
    \end{align*}
\end{lem}

\begin{proof}
    This is an immediate implication of Lemma A.2 of \citep{zhong2022pessimistic}.
\end{proof}


The next result provides upper bounds on the expected values of the feature norms, for all the trajectories followed by the policy pairs in the unilateral set of the NE policy pair. This result is based on the Low Relative Uncertainty assumption. 

\begin{lem}\label{lem:low_relative_uncertainty}
    Assume that Assumption \ref{asm:low_relative_uncertainty} holds, that is, assume that there exists a constant $c_1>0$ such that, for all $h\in[H]$, we have
    \begin{align*}
        \Lambda_h \succeq I + c_1 K\max \left\{ \sup_\nu \; \mathbb{E}_{\pi^*,\nu}[\phi_h\phi_h^\top \vert s_1=x], \;  \sup_\pi \; \mathbb{E}_{\pi,\nu^*}[\phi_h\phi_h^\top \vert s_1=x] \right\}~.
    \end{align*}
    Then, for all $h\in[H]$, with probability at least $1-\delta$, we have
    \begin{align*}
        \mathbb{E}_{\pi^*,\nu'}\left[  \norm{\phi(s,a,b)}_{\Lambda_h^{-1}}\right] +
        \mathbb{E}_{\pi',\nu^*}\left[  \norm{\phi(s,a,b)}_{\Lambda_h^{-1}}\right] \leq 2\sqrt{\frac{d}{c_1K}}~.
    \end{align*}
\end{lem}

\begin{proof}
    We derive an upper bound for the first expectation. The argument is identical to the second one. Let
    \begin{align}\label{eq:sigma_x_definition}
        \overline{\Sigma}_h(x) = \mathbb{E}_{\pi^*,\nu'} \left[ \phi(s_h,a_h,b_h)^\top\phi(s_h,a_h,b_h) \vert s_1=x \right]~.
    \end{align}
    First, note that, given two matrices $A, B \in \mathbb{R}^{n\times n}$ such that $A \preceq B$, then, for any $x\in\mathbb{R}^n_+$, we have $\norm{x}_A \leq \norm{x}_B$. This follows immediately by observing that, since $A-B \succeq 0$, which means that $A-B$ is positive semi-definite, we obtain $x^T(A-B)x \geq 0$, which implies the desired result. 
    Now, for any $h\in[H]$, we have
    \begin{align*}
        \mathbb{E}_{\pi^*,\nu'}\left[  \norm{\phi(s_h,a_h,b_h)}_{\Lambda_h^{-1}}\right] & = \mathbb{E}_{\pi^*,\nu'}\left[ \sqrt{\phi^\top_h \Lambda_h^{-1}\phi_h}  \right]\\
            & \leq \mathbb{E}_{\pi^*,\nu'}\left[ \sqrt{\phi^\top_h \left(c_1K\overline{\Sigma}_h(x) + I\right)^{-1}\phi_h} \big\vert s_1 = x \right]  \\
            & \leq \sqrt{ \mathbb{E}_{\pi^*,\nu'}\left[ Tr\left( \left(c_1K\overline{\Sigma}_h(x) + I\right)^{-1}\phi_h\phi^\top_h\right) \bigg\vert s_1 = x\right]} \\
            & = \sqrt{  Tr\left( \left(c_1K\overline{\Sigma}_h(x) + I\right)^{-1}\mathbb{E}_{\pi^*,\nu'}\left[\phi_h\phi^\top_h\big\vert s_1 = x\right]\right) } \\
            & =  \sqrt{ Tr\left( \left(c_1K\overline{\Sigma}_h(x) + I\right)^{-1}\overline{\Sigma}_h(x)\right)}\\
            & =  \sqrt{\frac{1}{c_1K} Tr\left( \left(c_1K\overline{\Sigma}_h(x) + I\right)^{-1} \left(\left(c_1K\overline{\Sigma}_h(x) + I\right) - I\right)\right)} \\
            & = \sqrt{\frac{1}{c_1K}} \sqrt{Tr\left( I - \left( c_1K\overline{\Sigma}_h(x) + I \right)^{-1} \right)} \\
            & \leq \sqrt{\frac{d}{c_1K}}~,
    \end{align*}
    where the first inequality follows by assumption; the second inequality follows from Jensen; the second equality follows from the fact that $Tr(\mathbb{E}X)=\mathbb{E}[Tr(X)]$, for a given random matrix $X$; for the fourth equality, we have used the observation that $Tr(I-A) \leq \dim (A)$, for a positive definite matrix $A$. 
\end{proof}

Now we are ready to prove Theorem \ref{thm:linear_case}. We restate it below for convenience.
\begin{statement}
    Suppose that Assumption \ref{asm:low_relative_uncertainty} holds with given constant $c_1$. Let $\delta >0$, $\epsilon < 1/2$ and let $D$ be the $\epsilon$-corrupted version of the dataset $\widetilde{D}$ comprised of $K$ trajectories of length $H$, where $K \geq \log (\min (K,d))/\epsilon$ and $(\widetilde{s}^\tau_h,\widetilde{a}^\tau_h,\widetilde{b}^\tau_h) = (s^\tau_h,a^\tau_h,b^\tau_h)$, for all $\tau \in [K]$, $h\in[H]$. Then, with probability at least $1-\delta$, SCRAM-PMVI outputs $(\widehat{\pi},\widehat{\nu})$ that satisfy, for every $s\in\mathcal{S}$:
    \begin{align*}
        \optgap (\widehat{\pi},\widehat{\nu},s) & \leq \widetilde{O}\left( \frac{1}{\sqrt{c_1}}(\gamma+H)H\sqrt{d}\epsilon + \frac{H^2d}{\sqrt{c_1K}}\right)~.
    \end{align*}
\end{statement}
We decompose the suboptimality gap as in Equation \eqref{eq:sub_opt_gap_decomposition}:
\begin{align*}
    \optgap (\widehat{\pi},\widehat{\nu},s) = V^{*,\widehat{\nu}}_1(s) - V^{\widehat{\pi},*}_1(s) = \left( V^{*,\widehat{\nu}}_1(s) - V^*_1(s)\right) + \left( V^*_1(s) - V^{\widehat{\pi},*}_1(s)\right) ~.
\end{align*}
First, we bound the left difference.
\begin{align*}
    V^{*,\widehat{\nu}}_1(s) - V^*_1(s) \leq \overline{V}_1(s) - V^*_1(s)  \leq \overline{V}_1(s) - V^{\pi',\nu^*}_1(s)~,
\end{align*}
for some $\pi'$, where the first inequality follows from Lemma \ref{lem:value_function_bounds} and the second inequality follows from the fact that $(\pi^*,\nu^*)$ is an NE policy. Similarly, we have that 
$V^*_1(s) - V^{\widehat{\pi},*}_1(s) \leq V^{\pi^*,\nu'}_1(s) - \underline{V}_1(s)$, for some $\nu'$. 
Then, similar to the proof of Theorem \ref{thm:linear_uniform_coverage}, we have
\begin{align*}
    V^{*,\widehat{\nu}}_1(s) - V^*_1(s) & \leq \overline{V}_1(s) - V^{\pi',\widehat{\nu}}_1(s) \\
    & \leq \sum^H_{h=1} \mathbb{E}_{\pi',\nu^*}\left[ -\overline{\iota}_h(s_h,a_h,b_h) \right] \\
        & \leq 2\sum^H_{h=1} \mathbb{E}_{\pi',\nu^*}\left[ \Gamma_h(s_h,a_h,b_h) \vert s_1=s \right] \\
        & \leq 2\left(\sqrt{K} \mathcal{E}(\epsilon,K,H,d) + 2H\sqrt{d} \right) \sum^H_{h=1} \mathbb{E}_{\pi',\nu^*}\left[  \norm{\phi(s,a,b)}_{\Lambda_h^{-1}}\right] 
\end{align*}
Similarly, we have
\begin{align*}
    V^*_1(s) - V^{\widehat{\pi},*}_1(s) & \leq 2\left(\sqrt{K} \mathcal{E}(\epsilon,K,H,d) + 2H\sqrt{d}  \right) \sum^H_{h=1} \mathbb{E}_{\pi^*,\nu'}\left[  \norm{\phi(s,a,b)}_{\Lambda_h^{-1}}\right] ~.
\end{align*}
Finally, applying Lemma \ref{lem:low_relative_uncertainty}, we obtain
\begin{align*}
    \optgap (\widehat{\pi},\widehat{\nu},s) & \leq 2\left( \sqrt{K} \mathcal{E}(\epsilon,K,H,d) + 2H\sqrt{d}  \right) \sum^H_{h=1} \left(  \mathbb{E}_{\pi',\nu^*}\left[  \norm{\phi(s,a,b)}_{\Lambda_h^{-1}}\right] + \mathbb{E}_{\pi^*,\nu'}\left[  \norm{\phi(s,a,b)}_{\Lambda_h^{-1}}\right] \right) \\
        & \leq 4\left( \sqrt{K} \mathcal{E}(\epsilon,K,H,d) + 2H\sqrt{d}  \right)  \sum^H_{h=1} \sqrt{\frac{d}{c_1K}} \\
        & = O \left( \frac{\gamma}{c_1}H\sqrt{d}\epsilon + H^2K^{-1/2}d\right) \\
        & \leq O\left( \frac{1}{c_1}(\gamma+H)H\sqrt{d}\epsilon + \frac{H^2d}{c_1\sqrt{K}}\right)~,
\end{align*}
where the last inequality follows from Lemma \ref{lem:variance_bound}.

\section{Proofs of Section \ref{sec:no_coverage_bounds}}
In this section, we derive the proofs of the results in Section \ref{sec:no_coverage_bounds}.

\subsection{Proof of Theorem \ref{thm:general_result}}

For this section, we prove a similar result which gives us general bounds when no coverage is guaranteed on the corrupted dataset, but the underlying clean data has LRU coverage.

First, since we assume corrupted covariates, we rewrite the overall covariance as
\begin{align}\label{eq:general_covariance}
    \Lambda_h = \widetilde{\Lambda}_h + \widehat{\Lambda}_h = \sum_{\text{clean} \; \tau} \widetilde{\phi}^\tau_h \left( \widetilde{\phi}^\tau_h\right)^\top + \sum_{\text{corrupted} \; \tau} \widehat{\phi}^\tau_h \left( \widehat{\phi}^\tau_h\right)^\top + I = (1-\epsilon)\left( K \widetilde{\Sigma}_h + I\right) +  \epsilon \left( K\widehat{\Sigma}_h + I\right)
\end{align}

We start by deriving upper bounds on the Bellman error. 

\begin{lem}\label{lem:bellman_error_no_coverage}
    With probability at least $1-\delta$, we have
    \begin{align*}
        0 & \leq \underline{\iota}_h(s,a,b) \leq 2 \widehat{\Gamma}_h(s,a,b)~, \\
        0 & \leq - \overline{\iota}_h(s,a,b) \leq 2 \widehat{\Gamma}_h(s,a,b)~,
    \end{align*}
    where 
    \begin{align*}
        \widehat{\Gamma}_h(s,a,b) = \left(\sqrt{(1-\epsilon)K}\mathcal{E} + (\sqrt{\epsilon K} + 2) H\sqrt{d} \right) \norm{\phi(s,a,b)}_{\Lambda^{-1}_h}
    \end{align*}
\end{lem}

\begin{proof}
    Let us consider the first part. The second part will follow by a similar argument. Similar to the argument of Lemma \ref{lem:bellman_error}, we have
    \begin{align*}
        \norm{\underline{\omega}_h - \underline{\omega}^*_h}^2_{\Lambda_h} & = (\underline{\omega}_h - \underline{\omega}^*_h)^\top \Lambda_h (\underline{\omega}_h - \underline{\omega}^*_h) \\
            & = (\underline{\omega}_h - \underline{\omega}^*_h)^\top \left( (1-\epsilon)\left( K \widetilde{\Sigma}_h + I\right) +  \epsilon \left( K\widehat{\Sigma}_h + I\right) \right) (\underline{\omega}_h - \underline{\omega}^*_h) \\
        & = (1-\epsilon)K \norm{\underline{\omega}_h - \underline{\omega}^*_h}^2_{\widetilde{\Sigma}_h} + \epsilon K \norm{\underline{\omega}_h - \underline{\omega}^*_h}^2_{\widehat{\Sigma}_h} + H^2 d \\
            & \leq (1-\epsilon)K \mathcal{E}^2 + \epsilon K \norm{\widehat{\Sigma}_h}_2 H^2d + 4H^2 d \\
        & \leq (1-\epsilon)K \mathcal{E}^2 + \epsilon K H^2d + 4H^2d~.
    \end{align*}

    Thus, we obtain
    \begin{align*}
        \norm{\underline{\omega}_h - \underline{\omega}^*_h}_{\Lambda_h} \leq \sqrt{(1-\epsilon)K}\mathcal{E} + (\sqrt{\epsilon K} + 2) H\sqrt{d}~.
    \end{align*}

Then, similar to Lemma \ref{lem:bellman_error}, with probability at least $1-\delta$, we have
\begin{align*}
   \vert \phi^\top\underline{\omega}_h - (\mathbb{B}_h\underline{V}_{h+1})(s,a,b)\vert & = \vert \phi^\top (\underline{\omega}_h - \underline{\omega}^*_h)\vert \nonumber 
         \leq \norm{\underline{\omega}_h - \underline{\omega}^*_h}_{\Lambda_h} \norm{\phi}_{\Lambda^{-1}_h}  
         \leq \left(\sqrt{(1-\epsilon)K}\mathcal{E} + (\sqrt{\epsilon K} + 2) H\sqrt{d} \right) \norm{\phi}_{\Lambda^{-1}_h}~.
\end{align*}
The rest of the proof follows similar arguments to Lemma \ref{lem:bellman_error}.
\end{proof}
Next, we derive upper bounds on the additional error term coming from the damage on coverage. 

\begin{lem}\label{lem:no_coverage_result}
    Assume that the condition of Lemma \ref{lem:low_relative_uncertainty} holds only for the underlying clean data, but not for the corrupted dataset. Then, we have
    \begin{align*}
        \mathbb{E}_{\pi^*,\nu'}\left[  \norm{\phi(s,a,b)}_{\Lambda_h^{-1}}\right] +
        \mathbb{E}_{\pi',\nu^*}\left[  \norm{\phi(s,a,b)}_{\Lambda_h^{-1}}\right] \leq  2\left( \sqrt{\frac{d}{c_1K}}+ \sqrt{\frac{d\epsilon}{c_1K(1-\epsilon)}} \right)~.
    \end{align*}
\end{lem}

\begin{proof}
    First, using the definition of $\Lambda_h$, we can equivalently write $\Lambda_h = S_h - \Delta_h$, where 
    \begin{align}\label{eq:S_h_def}
        S_h = \sum^K_{\tau =1} \widetilde{\phi}^\tau_h\left(\widetilde{\phi}^\tau_h\right)^\top + I
    \end{align}
    is the clean sample covariance matrix coming from $\widetilde{D}$, and $\Delta_h$ is defined as
    \begin{align}\label{eq:Delta_h_def}
        \Delta_h = \sum_{\textrm{corrupted } \tau} \tilde{\phi}_h^\tau \tilde{\phi}_h^{\tau^\top} - {\phi}_h^\tau {\phi}_h^{\tau^\top}~.
    \end{align}
    Since $\Delta_h$ is a symmetric matrix, we can write $\Delta_h = U_h E_h U_h^\top$, where $E_h$ is a diagonal matrix composed of the eigenvalues of $\Delta_h$. Note that the absolute value of each entry of $E_h$ is at most $\epsilon K$. 
    We will use the following formula for the inverse of a sum of two matrices \citep{henderson1981deriving}:
    $$
    (A + UBV)^{-1} = A^{-1} - A^{-1}U( I + BVA^{-1}U)^{-1} BVA^{-1}~,
    $$
    provided that $ I + BVA^{-1}U$ is non-singular. In contrast to the Woodbury matrix identity, the matrix $B$ here can be non-invertible. Using $A = S_h$ and $UBV = - U_h E_h U_h^\top$ we get the following identity.
    \begin{align}\label{eq:inversion_identity}
        \Lambda_h^{-1} = \left( S_h - \Delta_h\right)^{-1} = \left( S_h - U_hE_hU_h^\top\right)^{-1} =S_h^{-1} + S_h^{-1} \underbrace{U_h \left(I - E_h U_h^\top S_h^{-1} U_h \right)^{-1} E_h U_h^\top}_{:= M_h} S_h^{-1}~.
    \end{align}
This identity gives us the following bound on the bonus term with respect to $\pi^*$ and $\nu'$.
\begin{align}
        \mathbb{E}_{\pi^*,\nu'}\left[  \norm{\phi(s_h,a_h,b_h)}_{\Lambda_h^{-1}}\right] & = \mathbb{E}_{\pi^*,\nu'}\left[ \sqrt{\phi^\top_h \Lambda_h^{-1}\phi_h}  \right] \le \mathbb{E}_{\pi^*,\nu'}\left[ \sqrt{\phi^\top_h S_h^{-1}\phi_h}  \right] + \mathbb{E}_{\pi^*,\nu'}\left[ \sqrt{\left|\phi^\top_h S_h^{-1} M_hS_h^{-1}\phi_h\right|}  \right]
        \label{eq:decomposition-bonus-term-02}
\end{align}
The first term can be bounded by $\sqrt{\frac{d}{c_1 K}}$ by Lemma~\ref{lem:low_relative_uncertainty}, since it depends on the clean covariance matrix. For the second term, we first upper bound the maximum eigenvalue of the inverse of the middle expression in $M_h$. For that, we make use of the following Lemma.

\begin{lem}\label{lem:auxiliary_matrix_result}
    Given $\Sigma_h(x)$ as defined in Equation \eqref{eq:sigma_x_definition}, $S_h$ as defined in Equation \eqref{eq:S_h_def}, $\Delta_h=U_hS_hU_h^\top$ as defined in Equation \eqref{eq:Delta_h_def}, and $M_h$ as in Equation \eqref{eq:inversion_identity}, we have 
    \begin{align*}
        Tr\left( S^{-1}_hM_hS^{-1}_h\Sigma_h(x)\right) \leq \frac{d\epsilon}{c_1K(1-\epsilon)}~.
    \end{align*}
\end{lem}

\begin{proof}
 First, note that, since $U_h$ is an orthonormal matrix, we can write 
    \begin{align*}
        E_h U_h^\top S_h^{-1} U_h = E_h U_h^{-1} S_h^{-1} U_h = E_h \left(S_h U_h \right)^{-1} U_h~.
    \end{align*}
    Next, observe that
    \begin{align*}
        \norm{S_h U_h}_2 &= \max_{x\neq 0} \frac{\norm{U_h x + \sum^K_{\tau=1} \widetilde{\phi}_h^\tau (\widetilde{\phi}_h^\tau)^\top U_h x}_2}{\norm{x}_2} \\
            & \le \max_{x\neq 0} \frac{\norm{U_h x}_2 +  \norm{\sum^K_{\tau=1} \widetilde{\phi}_h^\tau (\widetilde{\phi}_h^\tau)^\top}_2\norm{U_h x}_2}{\norm{x}_2}\\
            & \leq \max_{x\neq 0}\frac{\norm{x}_2 + K \norm{x}_2}{\norm{x}_2}\\
        & = 1 + K~,
    \end{align*}
    where the first inequality follows by the triangle inequality and matrix norm properties, and for the second inequality we have used the fact that the maximum eigenvalue of the clean covariance matrix does not exceed $1+K$. This implies that $$\left(S_h U_h \right)^{-1} \succcurlyeq \frac{1}{1+K}I_d$$
    and, consequently, 
    \begin{align}\label{eq:M_h_lower_bound}
        I - E_h U_h^\top S_h^{-1} U_h \succcurlyeq \left(1 -\frac{\epsilon K}{1+K} \right) I_d~,
    \end{align}
    since the minimum eigenvalue of $-E_h$ is at least $-\epsilon K$. Thus, we obtain 
    \begin{align*}
        \left(I - E_h U_h^\top S_h^{-1} U_h\right)^{-1} \preceq \frac{1}{1 -\frac{\epsilon K}{1+K}}  I_d~ \preceq \frac{1}{1-\epsilon} I_d~.
    \end{align*}
    Using the argument above, we have
    \begin{align*}
        Tr\left( S^{-1}_h M_hS^{-1}_h\Sigma_h(x)\right) & = Tr \left(  S^{-1}_hU_h \left( I - E_hU_h^\top S^{-1}_h U_h\right)^{-1}E_hU_h^\top S^{-1}_hU_hU_h^\top \Sigma_h(x)\right)  \\
    & = Tr \left( S^{-1}_h U_h \left( \underbrace{\left( I - E_hU_h^\top S^{-1}_h U_h\right)^{-1}}_{A} - I \right) U_h^\top \Sigma_h(x) \right) \\
        & \leq \frac{1}{c_1K} Tr\left( U_h \left( A-I\right) U_h^\top \left( c_1K\Sigma_h(x) + I - I\right)\underbrace{\left(c_1K\Sigma_h(x) +I\right)^{-1}}_{B} \right) \\
    & = \frac{1}{c_1K} Tr\left( U_h \left( A-I\right) U_h^\top \left( I - B\right) \right) \\
        & \leq \frac{d}{c_1K}\norm{ U_h \left( A-I\right) U_h^\top \left( I - B\right)}_2 \\
    & \leq \frac{d}{c_1K}\norm{A-I}_2\norm{I-B}_2\\
        & \leq \frac{d}{c_1K} \left( \frac{1}{1-\epsilon} - 1\right)\\
    & = \frac{d\epsilon}{c_1K(1-\epsilon)}~,
    \end{align*}
    where the first inequality uses Assumption \ref{asm:low_relative_uncertainty} and the fact that $S^{-1}_h$ is symmetric; the second inequality uses the fact that $Tr(A)\leq d\norm{A}_2$; the third inequality uses the property $\norm{AB}_2\leq \norm{A}_2\norm{B}_2$, and the fact that, for orthonormal matrices, we have $\norm{U_h}_2\leq 1$.  Thus, by putting everything together, we obtain the desired bounds.
\end{proof}
Now, note that
\begin{align*}
    \expp_{\pi^*,\nu'}\left[ \sqrt{\phi^\top S^{-1}_hM_hS^{-1}_h\phi_h} \right]& \leq \sqrt{Tr\left( S^{-1}_hM_hS^{-1}_h\Sigma_h(x)\right)} \leq \sqrt{\frac{d\epsilon}{c_1K(1-\epsilon)}}~,
\end{align*}
by Lemma \ref{lem:auxiliary_matrix_result}. Thus, by putting everything together, we obtain
\begin{align*} \expp_{\pi^*,\nu'}\left[\norm{\phi(s_h,a_h,b_h)}_{\Lambda^{-1}_h}\right] \leq \sqrt{\frac{d}{c_1K}} +  \sqrt{\frac{d\epsilon}{c_1K(1-\epsilon)}}~.
\end{align*}
\end{proof}

Now we are ready to prove the main result of this section. We restate the result for convenience.
\begin{statement}
    Suppose that the condition of Assumption \ref{asm:low_relative_uncertainty} is satisfied only on the clean dataset $\widetilde{D}$, for a given constant $c_1$, and that an $\epsilon$-fraction of tuples in $D$ is arbitrarily corrupted. Furthermore, let $\delta > 0$ and $\epsilon \in (0,1/2)$, and $K \geq \log (\min \{ K,d\}) /\epsilon$. Then, with probability at least $1-\delta$, $S-PMVI$ returns $(\widehat{\pi},\widehat{\nu})$ that satisfy, for any $s\in\mathcal{S}$
    \begin{align*}
        \optgap(\widehat{\pi},\widehat{\nu},s) \leq O \left( \underbrace{\frac{1}{\sqrt{c_1K}}H^2d}_{\substack{\text{clean signal} \\ \text{\& clean covariates} \\ \text{\& LRU coverage}}} + \underbrace{\frac{1}{\sqrt{c_1}}H^2\sqrt{d}\epsilon}_{\substack{\text{corrupted signal}\\ \text{\& clean covariates} \\ \text{\& LRU coverage}}} + \underbrace{\frac{1}{\sqrt{c_1}}H^2d\sqrt{\epsilon}}_{\substack{\text{corrupted signal} \\ \text{\& corrupted covariates} \\ \text{\& LRU coverage}}} + \underbrace{\frac{1}{\sqrt{c_1(1-\epsilon)}}H^2d^{3/2}\epsilon}_{\substack{\text{corrupted signal} \\ \text{\& corrupted covariates}\\ \text{\& corrupted coverage}}} \right)
    \end{align*}
\end{statement}

\begin{proof}
    As in the proof of Theorem~\ref{thm:linear_case}, we have, for some policies $\pi'$ and $\nu'$:
    \begin{align}
    \optgap (\widehat{\pi},\widehat{\nu},s) & = V^{*,\widehat{\nu}}_1(s) - V^{\widehat{\pi},*}_1(s) = \left( V^{*,\widehat{\nu}}_1(s) - V^*_1(s)\right) + \left( V^*_1(s) - V^{\widehat{\pi},*}_1(s)\right) \nonumber \\
        & \leq \left( \overline{V}_1(s) - V^{\pi',\nu^*}_1(s)\right) + \left( V^{\pi^*,\nu'}_1(s) - \underline{V}_1(s)\right) \label{eq:no_coverage_proof_01}\\
    & \leq \sum^H_{h=1} \bigg( \mathbb{E}_{\pi',\nu^*}\left[ -\overline{\iota}_h(s_h,a_h,b_h)\right] + \mathbb{E}_{\pi^*,\nu'}\left[ \underline{\iota}_h(s_h,a_h,b_h) \right]  \bigg) \label{eq:no_coverage_proof_02}\\
        & \leq 2 \sum^H_{h=1} \bigg( \mathbb{E}_{\pi',\nu^*}\left[ \widehat{\Gamma}_h(s,a,b)|s_1=s\right] + \mathbb{E}_{\pi^*,\nu'}\left[ \widehat{\Gamma}_h(s,a,b)|s_1=s \right] \label{eq:no_coverage_proof_03} \bigg)\\
    & \leq 2 \left(\sqrt{(1-\epsilon)K}\mathcal{E} + (\sqrt{\epsilon K} + 2) H\sqrt{d} \right)\sum^H_{h=1} \left( \mathbb{E}_{\pi',\nu^*}\left[ \norm{\phi(s,a,b)}_{\Lambda^{-1}_h}\right] + \mathbb{E}_{\pi^*,\nu'}\left[ \norm{\phi(s,a,b)}_{\Lambda^{-1}_h} \right] \right) \nonumber \\
        & \leq 4H \left(\sqrt{(1-\epsilon)K}\mathcal{E} + \sqrt{\epsilon Kd}H + 2H\sqrt{d} \right)  \left( \sqrt{\frac{d}{c_1K}}+ \sqrt{\frac{d\epsilon}{c_1K(1-\epsilon)}} \right) \label{eq:no_coverage_proof_05}\\
    & \leq O \left( \underbrace{\frac{1}{\sqrt{c_1K}}H^2d}_{\substack{\text{clean signal} \\ \text{\& clean covariates} \\ \text{\& LRU coverage}}} + \underbrace{\frac{1}{\sqrt{c_1}}H^2\sqrt{d}\epsilon}_{\substack{\text{corrupted signal}\\ \text{\& clean covariates} \\ \text{\& LRU coverage}}} + \underbrace{\frac{1}{\sqrt{c_1}}H^2d\sqrt{\epsilon}}_{\substack{\text{corrupted signal} \\ \text{\& corrupted covariates} \\ \text{\& LRU coverage}}} + \underbrace{\frac{1}{\sqrt{c_1(1-\epsilon)}}H^2d\epsilon}_{\substack{\text{corrupted signal} \\ \text{\& corrupted covariates}\\ \text{\& corrupted coverage}}} \right) \nonumber \\
        & \leq O \left(  H^2d^{3/2}K^{-1/2} + H^2d^{3/2}\sqrt{\epsilon} \right) \label{eq:no_coverage_proof_06} ~,
\end{align}
where \eqref{eq:no_coverage_proof_01} follows from Lemma \ref{lem:uniform_coverage_value_function}; \eqref{eq:no_coverage_proof_01} follows from Equation \eqref{eq:sub_opt_gap_alternative}; \eqref{eq:no_coverage_proof_03} follows from Lemma \ref{lem:bellman_error}; \eqref{eq:no_coverage_proof_05} follows from Lemma \ref{lem:no_coverage_result} and \eqref{eq:no_coverage_proof_06} follows from the fact that, for $\epsilon \in (0,1/2)$, we have that $\epsilon / \sqrt{1-\epsilon} \leq \sqrt{\epsilon}$.
\end{proof}

\subsection{Proof of Theorem \ref{thm:second_general_result}}

In this section, we improve the order of $\epsilon$ in Theorem \ref{thm:general_result} under an additional assumption on the feature space (Assumption \ref{asm:features_lower_bound}.
We start by proving upper bounds on the Bellman error in terms of the bonus term. 

\begin{lem}\label{lem:improved_bellman_error}
    With probability at least $1-\delta$, we have
    \begin{align*}
        0 & \leq \underline{\iota}_h(s,a,b) \leq 2 \widehat{\Gamma}_h(s,a,b)~, \\
        0 & \leq -\overline{\iota}_h(s,a,b) \leq 2 \widehat{\Gamma}_h(s,a,b)~,
    \end{align*}
    for all $(s,a,b)\in \mathcal{S}\times \mathcal{A}\times \mathcal{B}$ and $h\in [H]$, where 
    \begin{align*}
        \widehat{\Gamma}_h(s,a,b) = \left( 2(1-\epsilon)K\mathcal{E} + \epsilon KH\sqrt{d} + H\sqrt{Kd}\right) \norm{\phi(s,a,b)^\top\Lambda^{-1}_h}_2~.
    \end{align*}
\end{lem}

\begin{proof}
    First, we express the Bellman error of any given $(s,a,b)$ tuple at time-step $h$ as
\begin{align*}
    \left| \overline{Q}_h(s,a,b) - (\mathbb{B}_h\overline{V}_{h+1})(s,a,b)\right| = \left| \phi^\top_h \left( \underline{\omega}_h - \underline{\omega}^*_h\right) \right| =  \left| \phi^\top_h \Lambda^{-1}_h \Lambda_h \left( \underline{\omega}_h - \underline{\omega}^*_h\right) \right| \leq \norm{\phi^\top_h \Lambda^{-1}_h}_2 \norm{\Lambda_h \left( \underline{\omega}_h - \underline{\omega}^*_h\right)}_2~.
\end{align*}
The error coming from the second player is similarly bounded. Note that we have
\begin{align*}
    & \norm{\Lambda_h \left( \underline{\omega}_h - \underline{\omega}^*_h\right)}^2_2  = \left|  \left( \underline{\omega}_h - \underline{\omega}^*_h\right)^\top \left( \widetilde{\Lambda}_h + \widehat{\Lambda}_h\right)^2  \left( \underline{\omega}_h - \underline{\omega}^*_h\right) \right| \\
        & = \left|  \left( \underline{\omega}_h - \underline{\omega}^*_h\right)^\top \left( \widetilde{\Lambda}^2_h + \widetilde{\Lambda}_h\widehat{\Lambda}_h +  \widehat{\Lambda}_h\widetilde{\Lambda}_h + \widehat{\Lambda}_h^2\right)  \left( \underline{\omega}_h - \underline{\omega}^*_h\right) \right| \\
    & \leq \underbrace{\left|  \left( \underline{\omega}_h - \underline{\omega}^*_h\right)^\top \widetilde{\Lambda}^2_h \left( \underline{\omega}_h - \underline{\omega}^*_h\right)\right|}_{P_1} + \underbrace{\left| \left( \underline{\omega}_h - \underline{\omega}^*_h\right)^\top \left( \widetilde{\Lambda}_h\widehat{\Lambda}_h +  \widehat{\Lambda}_h\widetilde{\Lambda}_h\right)  \left( \underline{\omega}_h - \underline{\omega}^*_h\right)\right|}_{P_2} + \underbrace{\left| \left( \underline{\omega}_h - \underline{\omega}^*_h\right)^\top \widehat{\Lambda}^2_h\left( \underline{\omega}_h - \underline{\omega}^*_h\right)\right|}_{P_3}
\end{align*}
We derive upper bounds for the three terms above separately. First, we have
\begin{align*}
    P_1 & = (1-\epsilon)\left| \left( \underline{\omega}_h - \underline{\omega}^*_h\right)^\top \left( K\widetilde{\Sigma}_h + I\right)\widetilde{\Lambda}_h \left( \underline{\omega}_h - \underline{\omega}^*_h\right) \right| \\
        & \leq (1-\epsilon) \norm{\widetilde{\Lambda}_h}_2 \left( K\norm{\underline{\omega}_h - \underline{\omega}^*_h}^2_{\widetilde{\Sigma}_h} + H^2d\right) \\
    & \leq (1-\epsilon)^2K^2 \mathcal{E}^2 + (1-\epsilon)^2K H^2d~,
\end{align*}
where the first equality follows from Equation \eqref{eq:general_covariance}, the first inequality follows from the fact that $\langle x, Ax\rangle \leq \norm{A}_2 \norm{x}^2_2$ and the last inequality follows from the fact that the maximum eigenvalue of $\widetilde{\Lambda}_h$ is at most $(1-\epsilon)K$, the fact that $\norm{\underline{\omega}^*_h}_2\leq H\sqrt{d}$ (Lemma E.1 of \cite{zhong2022pessimistic}) and Corollary \ref{cor:scram_guarantee}, where we deliberately omit the dependencies of $\mathcal{E}$ for brevity. 

Next, we have
\begin{align*}
    P_2 & \leq \left|  \left( \underline{\omega}_h - \underline{\omega}^*_h\right)^\top \widetilde{\Lambda}_h\widehat{\Lambda}_h  \left( \underline{\omega}_h - \underline{\omega}^*_h\right)\right| + \left|  \left( \underline{\omega}_h - \underline{\omega}^*_h\right)^\top \widehat{\Lambda}_h \widetilde{\Lambda}_h  \left( \underline{\omega}_h - \underline{\omega}^*_h\right)^\top\right|\\
        & \leq 2\norm{\widehat{\Lambda}_h}_2  \left( \underline{\omega}_h - \underline{\omega}^*_h\right)^\top\widetilde{\Lambda}_h  \left( \underline{\omega}_h - \underline{\omega}^*_h\right) \\
    & \leq 2\epsilon K \left(K \norm{\underline{\omega}_h - \underline{\omega}^*_h}^2_{\widetilde{\Sigma}_h} + H^2d\right) \\
        & = 2\epsilon(1-\epsilon) K^2 \mathcal{E}^2 + 2\epsilon(1-\epsilon)KH^2d~,
\end{align*}
where we have used the fact that the maximum eigenvalue of $\widehat{\Lambda}_h$ is at most $\epsilon K$ and applied similar arguments as for $P_1$. For the last term, we have
\begin{align*}
    P_3 & = \epsilon \left|  \left( \underline{\omega}_h - \underline{\omega}^*_h\right)^\top \left( K\widehat{\Sigma}_h +I\right)\widehat{\Lambda}_h  \left( \underline{\omega}_h - \underline{\omega}^*_h\right)\right| \\
        & \leq \epsilon\norm{\widehat{\Lambda}_h}_2  \left( K\norm{\underline{\omega}_h - \underline{\omega}^*_h}^2_2\norm{\widehat{\Sigma}_h}_2 + H^2d\right)\\
    & \leq \epsilon^2K \left( KH^2d + H^2d\right)\\
        & = \epsilon^2K^2H^2d + \epsilon^2KH^2d~,
\end{align*}
by applying the same arguments as above. Putting everything together, we obtain
\begin{align*}
    \norm{\Lambda_h \left( \underline{\omega}_h - \underline{\omega}^*_h\right)}^2_2 & \leq P_1 + P_2 + P_3\\
        & = (1-\epsilon)^2K^2 \mathcal{E}^2 + (1-\epsilon)^2K H^2d + 2\epsilon(1-\epsilon) K^2 \mathcal{E}^2 + 2\epsilon(1-\epsilon)KH^2d + \epsilon^2K^2H^2d + \epsilon^2KH^2d \\
    & =  (1-\epsilon)^2K^2 \mathcal{E}^2 + 2\epsilon(1-\epsilon) K^2 \mathcal{E}^2 + \epsilon^2K^2H^2d + \left( (1-\epsilon) H\sqrt{Kd} + \epsilon H\sqrt{Kd}\right)^2 \\
        & \leq 3(1-\epsilon)^2K^2 \mathcal{E}^2 +  \epsilon^2K^2H^2d + KH^2d~,
\end{align*}
where the last inequality follows from the fact that $\epsilon < 1/2$. Taking the square root of both sides and using the triangle inequality, we finally obtain
\begin{align*}
    \norm{\Lambda_h \left( \underline{\omega}_h - \underline{\omega}^*_h\right)}_2 \leq 2(1-\epsilon)K\mathcal{E} + \epsilon KH\sqrt{d} + H\sqrt{Kd}~.
\end{align*}
\end{proof}
Next, we derive an upper bound on the expected value of the feature norms with respect to policy pairs lying in the LRU set. 

\begin{lem}\label{lem:general_trace_bound}
    Assume that Assumption \ref{asm:low_relative_uncertainty} holds on the clean dataset $\widetilde{D}$ only, with given constant $c_1>0$. Then, for all $h\in[H]$, with probability at least $1-\delta$, we have
    \begin{align*}
        \expp_{\pi^*,\nu'}\left[ \norm{\phi(s_h,a_h,b_h)^\top\Lambda^{-1}_h}_2 \right] + \expp_{\pi',\nu^*}\left[ \norm{\phi(s_h,a_h,b_h)^\top\Lambda^{-1}_h}_2 \right]  \leq 2\left( \frac{d}{c_1c_2K} + \frac{d\epsilon}{c_1c_2K(1-\epsilon)}\right)~,
    \end{align*}
    where $c_2 \leq \min_{\phi \in \Phi} \norm{\phi}_2$. 
\end{lem}

\begin{proof} 
        Using the short-hand notation $\phi_h$ instead of $\phi(s_h,a_h,b_h)$, we have
    \begin{align*}
        \expp_{\pi^*,\nu'}\left[ \norm{\phi_h^\top\Lambda^{-1}_h}_2 \right]  & \leq \frac{1}{c_2} \expp_{\pi^*,\nu'}\left[ \sqrt{\phi^\top_h \left( S^{-1}_h + S^{-1}_h M_h S^{-1}_h \right)^2 \phi_h \phi^\top_h\phi_h}\right]\\
            & = \frac{1}{c_2} \expp_{\pi^*,\nu'}\left[\sqrt{  Tr \left( S^{-2}_h\left( I + M_hS^{-1}_h\right)^2 \left( \phi_h\phi^\top_h\right)^2\right)} \right] \\
        & \leq \frac{1}{c_2} \expp_{\pi^*,\nu'}\left[ Tr \left( S^{-1}_h\left( I + M_hS^{-1}_h\right) \phi_h\phi^\top_h\right) \right]\\
        & = \frac{1}{c_2} Tr\left(S^{-1}_h\left( I + M_hS^{-1}_h\right) \left( \Sigma_h(x) \right)\right) \\
            & \leq \frac{1}{c_2} \left(Tr\left( S^{-1}_h \Sigma_h(x)\right)+ Tr\left( S^{-1}_hM_hS^{-1}_h  \Sigma_h(x)\right)\right)~.
    \end{align*}
    The first factor is bounded as
    \begin{align*}
        Tr\left( S^{-1}_h\Sigma_h(x) \right) \leq \frac{d}{c_1K}
    \end{align*}
    by Lemma \ref{lem:low_relative_uncertainty}, while the second term is bounded as
    \begin{align*}
        Tr\left( S^{-1}_hM_hS^{-1}_h  \Sigma_h(x)\right) \leq \frac{d\epsilon}{c_1K(1-\epsilon)}~,
    \end{align*}
    by Lemma \ref{lem:auxiliary_matrix_result}. The result follows.
\end{proof}

\begin{statement}
    Suppose that the conditions of Theorem \ref{thm:general_result} and Assumption \ref{asm:features_lower_bound} hold. Then, with probability at least $1-\delta$, $S-PMVI$ returns $(\widehat{\pi},\widehat{\nu})$ that satisfy, for any $s\in\mathcal{S}$:
    \begin{align*}
        \optgap(\widehat{\pi},\widehat{\nu},s) \leq O \left( \frac{1}{c_1c_2}H^2d^{3/2}K^{-1/2} + \frac{1}{c_1c_2}H^2d^{3/2}\epsilon \right)~.
    \end{align*}
\end{statement}

\begin{proof}
Again, as in Theorem~\ref{thm:linear_case}, we have
    \begin{align}
    & \optgap (\widehat{\pi},\widehat{\nu},s) = V^{*,\widehat{\nu}}_1(s) - V^{\widehat{\pi},*}_1(s) = \left( V^{*,\widehat{\nu}}_1(s) - V^*_1(s)\right) + \left( V^*_1(s) - V^{\widehat{\pi},*}_1(s)\right) \nonumber \\
    & \leq \sum^H_{h=1} \bigg( \mathbb{E}_{\pi',\nu^*}\left[ -\overline{\iota}_h(s_h,a_h,b_h)\right] + \mathbb{E}_{\pi^*,\nu'}\left[ \underline{\iota}_h(s_h,a_h,b_h) \right]  \bigg) \nonumber\\
        & \leq 2 \sum^H_{h=1} \bigg( \mathbb{E}_{\pi',\nu^*}\left[ \widehat{\Gamma}_h(s,a,b)|s_1=s\right] + \mathbb{E}_{\pi^*,\nu'}\left[ \widehat{\Gamma}_h(s,a,b)|s_1=s \right] \bigg) \nonumber \\
    & \leq 2 \left( 2(1-\epsilon)K\mathcal{E} + \epsilon KH\sqrt{d} + H\sqrt{Kd}  \right)\sum^H_{h=1} \left( \mathbb{E}_{\pi',\nu^*}\left[ \norm{\phi(s,a,b)^\top\Lambda^{-1}_h}_2\right] + \mathbb{E}_{\pi^*,\nu'}\left[ \norm{\phi(s,a,b)^\top\Lambda^{-1}_h}_2 \right] \right) \label{eq:no_coverage_proof_11}\\
        & \leq 4H \left( 2(1-\epsilon)K\mathcal{E} + \epsilon KH\sqrt{d} + H\sqrt{Kd}  \right)\sum^H_{h=1} \left( \frac{d}{c_1c_2K} + \frac{d\epsilon}{c_1c_2K(1-\epsilon)} \right) \label{eq:no_coverage_proof_12}\\
        & \leq O \left(\frac{1}{c_1c_2}  H^2d^{3/2}K^{-1/2} + \frac{1}{c_1c_2} H^2d^{3/2}\epsilon \right) \nonumber ~,
    \end{align}
where \eqref{eq:no_coverage_proof_12} follows from Lemma \ref{lem:improved_bellman_error}, and \eqref{eq:no_coverage_proof_12} follows from Lemma \ref{lem:general_trace_bound}~.
\end{proof}

\section{Proofs of Section \ref{sec:coverage_discussion}}\label{sec:assumptions_strength}

In this section, we provide the proofs of results related to coverage assumptions in Section \ref{sec:coverage_discussion}.

\subsection{Proof of Proposition \ref{prop:equivalent_assumptions}}

 Before proving Proposition \ref{prop:equivalent_assumptions}, we state a result that provides bounds on the concentration of covariance matrices. For proof see \citep{zanette2021cautiously}. 
\begin{lem}\label{lem:covariance_concentration}
    Let $\{ \phi_i\}_{i\in[K]} \subset \mathbb{R}^d$ be i.i.d. samples from an underlying bounded distribution $\nu$, with $\norm{\phi_i}_2\leq 1$, for $i\in[K]$ and covariance $\Sigma$. Define
    \begin{equation*}
        \Lambda = \sum^K_{k=1}\phi_i\phi_i^\top + \lambda I~,
    \end{equation*}
    where $\lambda \geq \Omega ( d\log (K/\delta))$. Then, with probability at least $1-\delta$, we have
    \begin{equation*}
        \frac{1}{3}(K\Sigma + \lambda I) \preceq \Lambda \preceq \frac{5}{3}(K\Sigma + \lambda I)~.
    \end{equation*}
\end{lem}

 We restate Proposition \ref{prop:equivalent_assumptions}.
 \begin{statement}
     Assume $\Phi$ has full rank and let $\delta\in(0,1)$. Then, if Assumption \ref{asn:unilateral-coverage} holds, there exists a positive constant that depends on $\delta$ for which  Assumption \ref{asm:low_relative_uncertainty} holds, with probability at least $1-\delta$. Moreover, in the tabular MG setting, these two assumptions are equivalent. 
 \end{statement}

 \begin{proof}
      Fix $x\in\mathcal{S}$. Given policy pair $(\pi,\nu)$ and $h\in[H]$, let the matrix $D^{\pi,\nu}_h\in \mathbb{R}^{SAB \times SAB}$ denote the diagonal matrix composed of $d^{\pi,\nu}_h(s,a,b)$ diagonal entries, where we set $d^{\pi,\nu}_1(x,a,b)=1$. Note that 
     \begin{equation}
         \mathbb{E}_{\pi,\nu}\left[ \phi(s_h,a_h,b_h)\phi(s_h,a_h,b_h)^\top \vert s_1=x \right] = \Phi^\top D^{\pi,\nu}_h\Phi~.
     \end{equation}
     Further, let us denote by $\rho$ the behavioral policy from which the clean tuples $(s,a,b)$ from the given offline data are sampled and let us define the diagonal matrix $D^\rho_{\tau,h}\in \mathbb{R}^{SAB\times SAB}$, for every $\tau\in[K], h\in[H]$, with diagonal entries
     \begin{equation*}
         \widetilde{D}^\rho_{\tau,h} [(s,a,b)]:= \widetilde{d}^\rho_{\tau,h}(s,a,b) =
         \begin{cases}
             1 \;\; \text{if} \; (s^\tau_h,a^\tau_h,b^\tau_h) = (s,a,b)\\
             0 \;\; \text{otherwise}
         \end{cases}
     \end{equation*}
    and let 
    \begin{equation*}
        \widehat{D}^\rho_h = \frac{1}{K}\sum^K_{\tau=1}\widetilde{D}^\rho_{\tau,h}
    \end{equation*}
    denote the sample covariance matrix. By assumption, there exists a finite positive constant $c$ such that
    \begin{equation*}
          \frac{d^{\pi,\nu}_h(s,a,b)}{d^\rho_h(s,a,b)} < c < \infty, \forall h\in[H], (\pi,\nu)\in\mathcal{U}(\pi^*,\nu^*), (s,a,b)\in\{ (s,a,b) \in \mathcal{S}\times\mathcal{A}\times\mathcal{B} : d^{\pi,\nu}_h(s,a,b)>0\}~,
     \end{equation*}
     where $\mathcal{U}(\pi^*,\nu^*) = \{ (\pi^*,\nu),(\pi,\nu^*), \forall \pi, \forall \nu \}$. This implies that 
     \begin{equation}\label{eq:equivalent_assumptions_01}
         \min_{(s,a,b)\in\mathcal{S}\times\mathcal{A}\times\mathcal{B}} \left( d^\rho_h(s,a,b) - \frac{1}{c}d^{\pi,\nu}_h(s,a,b)\right) \geq 0~,
     \end{equation}
     which implies that 
     \begin{equation}\label{eq:unilateral_rlu_02}
         D^\rho_h - \frac{1}{c}D^{\pi,\nu}_h \succeq 0~.
     \end{equation}
     Now, Lemma \ref{lem:covariance_concentration} implies that there exist constants $\lambda_h \geq \Omega (d\log(KH/\delta))$, for $h\in[H]$, such that, with probability at least $1-\delta$, we have 
     \begin{equation*}
         \frac{1}{3}(KD^\rho_h + \lambda_h I ) \preceq \sum^K_{k=1} \widetilde{D}^\rho_{k,h} + \lambda_h I \preceq \frac{5}{3}( KD^\rho_h+ \lambda_h I)~,
     \end{equation*}
     for all $h\in[H]$. This implies that 
     \begin{equation}\label{eq:unilateral_rlu_03}
         \frac{1}{3}D^\rho_h - \frac{2\lambda_h}{3K}I \preceq \widehat{D}^\rho_h \preceq \frac{5}{3}D^\rho_h + \frac{2\lambda_h}{3K}I~.
     \end{equation}
     Equations \eqref{eq:unilateral_rlu_02} and \eqref{eq:unilateral_rlu_03} imply that, with probability at least $1-\delta$, we have
     \begin{equation*}
         \widehat{D}^\rho_h - \left( \frac{1}{c}D^{\pi,\nu}_h - \frac{2\lambda_h}{3K}I \right) \succeq 0~.
     \end{equation*}
     Since the left-hand side is a diagonal matrix, the above is equivalent to
     \begin{equation*}
         \min_{(s,a,b)} \left( \widehat{d}^\rho_h(s,a,b) - \left( \frac{1}{c} - \frac{2\lambda_h}{3Kd^{\pi,\nu}_h(s,a,b)} \right) d^{\pi,\nu}_h(s,a,b)\right) \geq 0~.
     \end{equation*}
     Now let us define constant $c_1$ as 
     \begin{align*}
         c_1 = \max \{ \frac{1}{c} - \frac{2d\log(KH/\delta)}{3Kd^{\pi,\nu}_h(s,a,b)} : d^{\pi,\nu}_h(s,a,b) >0\}~.
     \end{align*}
     Then we obtain
     \begin{equation*}
         \min_{(s,a,b)} \left( \widehat{d}^\rho_h(s,a,b) - c_1d^{\pi,\nu}_h(s,a,b)\right) \geq 0,
     \end{equation*}
     which means that we have
     \begin{equation*}
         \widehat{D}^\rho_h - c_1D^{\pi,\nu}_h \succeq 0~.
     \end{equation*}
     Now, since $\Phi$ has full rank, its null space is $0$. Thus, given non-zero $x\in\mathbb{R}^d$, we have $\Phi x \neq 0$ and thus
     \begin{equation*}
         (\Phi x)^\top \left( \widehat{D}^\rho_h - c_1D^{\pi,\nu}_h \right) (\Phi x) \geq 0~,
     \end{equation*}
     since $\widehat{D}^\rho_h - c_1D^{\pi,\nu}_h \succeq 0$. Therefore, 
     \begin{equation*}
         \Phi^\top \left( \widehat{D}^\rho_h - c_1D^{\pi,\nu}_h \right) \Phi \succeq 0~,
     \end{equation*}
     which implies
     \begin{equation*}
        \sum^K_{\tau=1}  \Phi^\top \widetilde{D}^\rho_{\tau,h}  \Phi  + I \succeq I + c_1K \Phi^\top D^{\pi,\nu}_h\Phi~.
     \end{equation*}
     Note that the above can be written as
     \begin{equation*}
         \Lambda_h \succeq I + c_1 K \max \left\{ \sup_\nu \; \mathbb{E}_{\pi^*,\nu}\left[ \phi_h\phi^\top_h \vert s_1=x \right], \; \sup_\nu \; \mathbb{E}_{\pi,\nu^*}\left[ \phi_h\phi^\top_h \vert s_1=x \right] \right\}~.
     \end{equation*}

     Now assume the MG is tabular. Then $\Phi$ is just the identity matrix. Thus, if Assumption \ref{asm:low_relative_uncertainty} holds, then there exists a positive constant $c_1$ such that, for any $h\in[H]$, $(\pi,\nu)\in \mathcal{U}(\pi^*,\nu^*)$, we have
     \begin{align*}
         \Lambda_h \succeq I + c_1K \Phi^\top D^{\pi,\nu}_h\Phi & \Rightarrow \Phi^\top \left( \widehat{D}^\rho_h - c_1 D^{\pi,nu}_h\right) \Phi \succeq 0  \\
                & \Rightarrow \widehat{D}^\rho_h - c_1 D^{\pi,\nu}_h \succeq 0\\
                & \Rightarrow \frac{5}{3}D^\rho_h + \frac{2\lambda_h}{3K}I - c_1D^{\pi,\nu}_h \succeq 0 \\
                & \Rightarrow  D^\rho_h + \frac{2\lambda_h}{5K}I - \frac{3}{5}c_1D^{\pi,\nu}_h \succeq 0 \\
                & \Rightarrow \min_{(s,a,b)} \left( d^\rho_h(s,a,b) - \left(\frac{3}{5}c_1 - \frac{2\lambda_h}{3Kd^{\pi,\nu}_h(s,a,b)}\right) d^{\pi,\nu}_h(s,a,b) \right) \geq 0~.
     \end{align*}
     Now let us define 
     \begin{equation*}
         c' = \max \left\{ \frac{3}{5}c_1 - \frac{2\lambda_h}{5Kd^{\pi,\nu}_h(s,a,b)}: d^{\pi,\nu}_h(s,a,b)>0 \right\}~.
     \end{equation*}
     We have 
     \begin{equation*}
          \min_{(s,a,b)} \left( d^\rho_h(s,a,b) -c' d^{\pi,\nu}_h(s,a,b) \right) \geq 0~.
     \end{equation*}
     Thus, for $c=1/c'$, we obtain 
     \begin{equation*}
          \frac{d^{\pi,\nu}_h(s,a,b)}{d^\rho_h(s,a,b)} < c < \infty, \forall h\in[H], (\pi,\nu)\in\mathcal{U}(\pi^*,\nu^*), (s,a,b)\in\{ (s,a,b) \in \mathcal{S}\times\mathcal{A}\times\mathcal{B} : d^{\pi,\nu}_h(s,a,b)>0\}~.
     \end{equation*}
\end{proof}

\subsection{Proof of Proposition \ref{prop:strength_of_assumptions}}

First, we show the equivalence of both assumptions in the tabular Markov game setting. Note that, in the tabular setting, the features are $SAB$ dimensional and $\phi(s,a,b)$ is the unit vector $e_{s,a,b}$ with coordinate $(s,a,b)$ set to $1$. First, suppose Assumption \ref{asn:uniform-phi-coverage} is true. Then we have,
\begin{equation}\label{eq:tabular-coverage}
 \E_{d^\rho_h}\left[ e_{s,a,b} e_{s,a,b}^\top \right] = \textrm{diag}\left[ \{d^\rho_h(s,a,b)\}_{s,a,b}\right] \succeq \xi \Identity
 \end{equation}
 This implies that $d^\rho_h(s,a,b) \ge \xi$ for any $h$ and $s,a,b$ and Assumption \ref{asn:uniform-coverage} is satisfied.
 
 Now, suppose Assumption \ref{asn:uniform-coverage} is true.
 Then $d^\rho_h(s,a,b) > 0$ for any $h$ and tuple $(s,a,b)$. Since the number of states is finite, there exists a constant $C$ such that $d^\rho_h(s,a,b) \ge C$ for any $h$ and any tuple $(s,a,b)$. Therefore, Equation~\eqref{eq:tabular-coverage} is satisfied with $\xi = C$, and thus, Assumption \ref{asn:uniform-phi-coverage} is satisfied.
 
 Next, we show that Assumption~\ref{asn:uniform-phi-coverage} is actually a stronger assumption than Assumption~\ref{asn:uniform-coverage} for the more general linear MDP model. We will write $\Phi \in \R^{SAB \times d}$  to denote the feature matrix where the row $(s,a,b)$ corresponds to the $d$-dimensional feature $\phi(s,a,b)$.\footnote{For infinite $\calS$, $\Phi$ is interpreted as a function.}
 \begin{lem}\label{lem_strong_assumption}
 Suppose $\textrm{rank}(\Phi) = d$. Then Assumption~\ref{asn:uniform-phi-coverage} implies Assumption~\ref{asn:uniform-coverage}.
 \end{lem}
 \begin{proof}
 We will assume $\calS$ is finite but possibly very large. The proof can be easily generalized  for infinite $\calS$. Since $\Phi$ has rank $d$ let us write $\Phi = U \Lambda V^\top$ where $U\in \R^{SAB \times d}$, $\Lambda$ is a $d$-dimensional diagonal matrix, and $V \in \R^{d \times d}$. Moreover, we can take $V$ to be orthonormal i.e. $V^\top V = \Identity$. Let $D^\rho_h \in \R^{SAB \times SAB}$ be a diagonal matrix with $D^\rho(s,a,b) = d^\rho_h(s,a,b)$.
 \begin{align*}
     &\E_{d^\rho_h} \left[ \phi(s,a,b) \phi(s,a,b)^\top \right] = \Phi^\top D^\rho_h \Phi \\
     &= V \Lambda U^\top D^\rho_h U \Lambda V^\top \succeq \xi \Identity_{d \times d}
 \end{align*}
 Since $\Phi$ has rank $d$, both $\Lambda$ and $V$ are invertible. This gives us 
\begin{align*}
 &U^\top D^\rho_h U \succeq \xi (V\Lambda)^{-1} (\Lambda V^\top)^{-1} \\
 &= \xi \Lambda^{-1} (V^\top V)^{-1} \Lambda^{-1} \succeq \xi \Lambda^{-2} \Identity_{d \times d}
 \end{align*}
 Since the matrix $U^\top D^\rho_h U$ is already in diagonalized representation, the minimum eigenvalue of $U^\top D^\rho_h U$ is the smallest diagonal entry of $D^\rho_h$. Therefore,
 $$
 \lambda_{\min} \left(U^\top D^\rho_h U \right) = \min_{s,a,b} d^\rho_h(s,a,b) \ge \frac{\xi}{\max_{j\in [d]} \Lambda(j)^2}
 $$
 Since $\Phi$ has rank $d$ at least one entry of the diagonal matrix $\Lambda$ is non-zero. Therefore, for any $h$, the tuple $(s,a,b)$ is covered with probability at least $p$ where $p = \xi / \max_j \Lambda(j)^2$.
 
 On the other hand, note that, for general rank-$d$ feature matrix $\Phi$, assumption~\ref{asn:uniform-coverage} need not imply assumption~\ref{asn:uniform-phi-coverage}. However, if we put additional restrictions on the features e.g. diversity, these two assumptions could be equivalent.
 \end{proof}

\end{document}